%% file: centrality.tex
\pdfoutput=1

\documentclass[letterpaper,11pt]{article}

\usepackage[normalem]{ulem}
\usepackage{graphicx}
\usepackage[latin9]{inputenc}
\usepackage[T1]{fontenc}
\usepackage{lmodern}

\usepackage{microtype}

\usepackage{amsmath}
\usepackage{amssymb}
\usepackage{amsthm}
\usepackage{thmtools}
\usepackage{thm-restate}

\usepackage{titlesec}
\usepackage{fancybox}
\usepackage{xcolor}
\usepackage{xspace}
\usepackage{enumerate}
\usepackage[english]{babel}

\usepackage{comment}

\usepackage{appendix}

\usepackage{bm}

\usepackage{fullpage}
\usepackage[procnumbered,ruled,vlined,linesnumbered]{algorithm2e}

\usepackage[colorlinks=true]{hyperref}

\newtheorem{theorem}{Theorem}[section]

\newtheorem{lemma}[theorem]{Lemma}

\newtheorem{fact}[theorem]{Fact}

\theoremstyle{definition}
\newtheorem{definition}[theorem]{Definition}
\newtheorem{remark}[theorem]{Remark}
\newtheorem{problem}{Problem}

\newenvironment{fminipage}%
{\begin{Sbox}\begin{minipage}}%
		{\end{minipage}\end{Sbox}\fbox{\TheSbox}}

\def\expec#1#2{{\mathbb{E}}_{#1}\left[ #2 \right]}

\def\defeq{\stackrel{\mathrm{def}}{=}}
\def\setof#1{\left\{#1  \right\}}
\def\sizeof#1{\left|#1  \right|}

\def\trace#1{\mathrm{Tr} \left(#1 \right)}

\def\floor#1{\left\lfloor #1 \right\rfloor}
\def\ceil#1{\left\lceil #1 \right\rceil}

\def\abs#1{\left|#1  \right|}

\def\norm#1{\left\| #1 \right\|}
\def\len#1{\left\lVert #1 \right\rVert}
\def\kh#1{\left( #1 \right)}

\newcommand{\SStil}{\boldsymbol{\mathit{\widetilde{S}}}}

\newcommand{\matlowtil}{\boldsymbol{\mathit{\widetilde{\mathcal{L}}}}}
\newcommand{\calDtil}{\boldsymbol{\mathit{\widetilde{\mathcal{D}}}}}

\newcommand{\calD}{\boldsymbol{\mathit{\mathcal{D}}}}

\newcommand{\matlow}{\boldsymbol{\mathit{{\mathcal{L}}}}}

\newcommand{\PP}{\boldsymbol{\mathit{P}}}

\newcommand{\rea}{\mathbb{R}}

\newcommand{\eps}{\epsilon}

\newcommand{\vecone}{\boldsymbol{\mathit{1}}}

\newcommand{\matzero}{\boldsymbol{\mathit{0}}}

\newcommand\PPi{\boldsymbol{\Pi}}

\newcommand\cchi{\boldsymbol{\chi}}

\newcommand\bb{\boldsymbol{\mathit{b}}}
\newcommand\cc{\boldsymbol{\mathit{c}}}

\newcommand\ee{\boldsymbol{\mathit{e}}}

\newcommand\qq{\boldsymbol{\mathit{q}}}

\newcommand\vv{\boldsymbol{\mathit{v}}}

\newcommand\yy{\boldsymbol{\mathit{y}}}
\newcommand\zz{\boldsymbol{\mathit{z}}}
\newcommand\xx{\boldsymbol{\mathit{x}}}

\renewcommand\AA{\boldsymbol{\mathit{A}}}
\newcommand\BB{\boldsymbol{\mathit{B}}}
\newcommand\CC{\boldsymbol{\mathit{C}}}
\newcommand\DD{\boldsymbol{\mathit{D}}}

\newcommand\II{\boldsymbol{\mathit{I}}}

\newcommand\NN{\boldsymbol{\mathit{N}}}

\newcommand\LL{\bm{\mathit{L}}}
\newcommand\LH{\bm{\mathit{H}}}
\newcommand\QQ{\bm{\mathit{Q}}}

\renewcommand\SS{\boldsymbol{\mathit{S}}}

\newcommand\WW{\boldsymbol{\mathit{W}}}
\newcommand\VV{\boldsymbol{\mathit{V}}}

\newcommand\LLtil{\boldsymbol{\mathit{\tilde{L}}}}

\newcommand\ZZ{\boldsymbol{\mathit{Z}}}

\makeatletter

\def\moverlay{\mathpalette\mov@rlay}
\def\mov@rlay#1#2{\leavevmode\vtop{%
		\baselineskip\z@skip \lineskiplimit-\maxdimen
		\ialign{\hfil$\m@th#1##$\hfil\cr#2\crcr}}}
\newcommand{\charfusion}[3][\mathord]{
	#1{\ifx#1\mathop\vphantom{#2}\fi
		\mathpalette\mov@rlay{#2\cr#3}
	}
	\ifx#1\mathop\expandafter\displaylimits\fi}
\makeatother

\colorlet{DarkRed}{red!50!black}
\colorlet{DarkGreen}{green!50!black}
\colorlet{DarkBlue}{blue!50!black}

\hypersetup{
	linkcolor = DarkRed,
	citecolor = DarkGreen,
	urlcolor = DarkBlue,
	bookmarksnumbered = true,
	linktocpage = true
}

\DontPrintSemicolon
\SetKw{KwAnd}{and}
\SetProcNameSty{textsc}
\SetFuncSty{textsc}
\SetKwInOut{Input}{Input\ \ \ \ }
\SetKwInOut{Output}{Output}

\let\oldnl\nl
\newcommand{\nonl}{\renewcommand{\nl}{\let\nl\oldnl}}

\ifdefined\ShowComment

\fi

\newcommand{\bsk}{\backslash_\theta}
\newcommand{\Lke}{\kh{ \LL\backslash_\theta e }}

\newcommand{\SC}{\textsc{Sc}}
\newcommand{\ECComp}{\textsc{EdgeCentComp1}}
\newcommand{\ECComptwo}{\textsc{EdgeCentComp2}}
\newcommand{\VCComp}{\textsc{VertexCentComp}}
\newcommand{\exactQuad}{\textsc{ExactQuad}}
\newcommand{\quadEst}{\textsc{QuadEst}}
\newcommand{\quadApx}{\textsc{QuadApprox}}

\newcommand{\ChebSolve}{\textsc{ChebSolve}}
\newcommand{\QuadSolver}{\textsc{QuadSolve}}
\newcommand{\LaplSolver}{\textsc{LaplSolve}}
\newcommand{\Kf}[1]{\mathcal{K}\kh{#1}}
\newcommand{\Ev}[1]{E_{#1}}
\newcommand{\degu}[1]{\mathrm{deg}^{\mathrm{u}}(#1)}

\newcommand{\sm}[2]{_{#1#2}}

\newcommand{\partialChol}{\textsc{ApxPartialCholesky}}

\newcommand{\EQ}{E^Q}
\newcommand{\ncur}{n_{\mathrm{cur}}}
\newcommand{\mcur}{m_{\mathrm{cur}}}

\DeclareRobustCommand{\vect}[1]{\bm{\mathit{#1}}}
\title{Kirchhoff Index As a Measure of Edge Centrality in Weighted Networks: Nearly Linear Time Algorithms}
\author{}

\author{
	Huan Li \\
	\normalsize School of Computer Science \\
	\normalsize Fudan University\\
	\normalsize \texttt{huanli16@fudan.edu.cn}
	\and
	Zhongzhi Zhang \\
	\normalsize School of Computer Science \\
	\normalsize Fudan University\\
	\normalsize \texttt{zhangzz@fudan.edu.cn}
}

\date{}

\usepackage[margin=1in]{geometry}
\usepackage{setspace}

\begin{document}
{\singlespacing

\pagenumbering{roman}

\setcounter{page}{0}
\pagenumbering{arabic}

\maketitle

\begin{abstract}
Estimating the relative importance of vertices and edges is a fundamental issue in the analysis of complex networks, and has found vast applications in various aspects, such as social networks, power grids, and biological networks. Most previous work focuses on metrics of vertex importance and methods for identifying powerful vertices, while related work for edges is much lesser, especially for weighted networks, due to the computational challenge.  In this paper, we propose to use the well-known Kirchhoff index as the  measure of edge centrality in weighted networks, called $\theta$-Kirchhoff edge centrality. The  Kirchhoff index of a network is defined as the sum of effective resistances over all vertex pairs. The centrality of an edge $e$ is reflected in the increase of Kirchhoff index  of the network when the edge $e$  is partially deactivated, characterized by a parameter $\theta$. We define two equivalent measures for $\theta$-Kirchhoff edge centrality. Both are  global metrics and have a better discriminating power than commonly used measures, based on local or partial structural information of networks, e.g. edge betweenness and spanning edge centrality.

Despite the strong advantages of Kirchhoff index as a centrality measure and its wide applications, computing the exact value of  Kirchhoff edge centrality for each edge in a graph is computationally demanding.  To solve this problem, for each of the  $\theta$-Kirchhoff edge centrality metrics, we present an efficient algorithm to compute its $\eps$-approximation for all the $m$ edges in nearly linear time in $m$. The proposed $\theta$-Kirchhoff edge centrality is the first  global metric of  edge importance that can be provably approximated in nearly-linear time.   Moreover, according to the  $\theta$-Kirchhoff edge centrality, we present a $\theta$-Kirchhoff vertex centrality measure, as well as a fast algorithm that can compute $\eps$-approximate Kirchhoff vertex centrality for all the $n$ vertices in nearly linear time in $m$.

\vspace{0.5cm}


\end{abstract}

\newpage

\section{Introduction}

Most real networks (e.g. social networks) are massive and inhomogeneous~\cite{Ne10}, where the roles of vertices/edges are often largely different, with peripheral vertices/edges having a limited effect on their function, while central vertices/edges having a strong impact on dynamical processes. Thus, it is of paramount importance to design both desirable metrics measuring the centrality or importance of vertices/edges and fast algorithms identifying vital vertices/edges~\cite{LaMe12}. In past decades, a lot of centrality measures have been presented to capture diverse aspects of the informal concept of importance, and various algorithms for these different metrics have been developed by researchers  from interdisciplinary areas, such as computer science~\cite{WhSm03,BoVi14,BoDeRi16}, control science~\cite{YoTeQi17}, and physics~\cite{LuChReZhZhZh16}. At present, it is still an active research topic in the scientific community.

Most previous work about centrality measures and algorithms concentrated on the vertex level, in spite of the fact that edge centrality plays an equally important role as its vertex counterpart. For example, edge centrality has been applied to detect communities of a network~\cite{GiNe02}, which are dense subgraphs corresponding to functional units within the network. In addition, edge centrality is helpful to describe the intensity of social ties among individuals in social networks~\cite{DiStKiPaBr11}, and is also instrumental in revealing new knowledge in semantic web~\cite{BeHeLa01}. Last but not the least, edge centrality plays an indispensable role in designing or protecting infrastructure networks, e.g. power grids~\cite{BiChH14}. Therefore, it is interesting to propose an edge centrality measure and develop algorithms characterizing the importance of an edge in networks relative to other edges.

Several measures for edge centrality have been proposed, including edge betweenness~\cite{Br01,BaKiMaMi07,BrPi07,GeSaSc08}, spanning edge centrality~\cite{TeMoCaRaFa13, MaGaKoTe15, HaAkYo16}, and current-flow centrality~\cite{BrFl05}. The betweenness of an edge is the fraction of shortest paths between vertex pairs that pass through the edge. The spanning edge centrality of an edge is defined as the probability that it is present in a randomly chosen spanning tree. While the current-flow centrality of an edge describes the amount of current flowing through it. Although these edge centrality metrics have been extensively studied, they themselves are subject to weakness. For example, edge betweenness only considers shortest paths and ignores those longer paths; spanning edge centrality cannot separate an edge linked to a leaf vertex and another cut edge connecting two large subgraphs, while the importance of these two edges are obviously different. Moreover, these measures are proposed for unweighted networks, and are either unapplicable to weighted networks, or have high computational complexity when applied to weighted networks.

In fact, it is very difficult to rigorously compare different measures of edge centrality, since the criteria of edge importance depend on real applications and the problems we are concerned with~\cite{YoTeQi17}. Hence, it is neither practical nor feasible to propose a universal measure that best quantifies the importance of edges for all situations. One should define the metric of edge centrality according to particular problems. In many real scenarios,  the Kirchhoff index~\cite{KlRa93} of a network, defined as the  sum of effective resistances over all vertex pairs, can be used as a unifying indicator to measure the interesting quantities associated with different  problems in networks. For example, Kirchhoff index can be used to measure the mean cost of search in a complex network~\cite{FeQuYi14}, robustness of first-order consensus algorithm in noisy networks~\cite{PaBa14}, the global utility of social recommender systems~\cite{WoLiCh16}, among others. Notwithstanding the relevance of the Kirchhoff index in various applications, there is a disconnect between this notion  and efficiency of algorithms for estimating it.

The main purpose of this paper  is to develop an edge centrality notion that not only has good discriminating power, but also can be evaluated using algorithms with good provable performances. For a connected graph, the popular Kirchhoff index follows Rayleigh's monotonicity law~\cite{ElSpVaJaKo11}. That is, the Kirchhoff index of a graph strictly increases when the weight of any edge is decreased. Based on this property, in this paper, we adopt the Kirchhoff index as an importance measure of edge in undirected weighted connected networks with a positive weight for each edge, which we call Kirchhoff edge centrality. In order to explore the role of an edge $e$, we partially deactivate the edge $e$ by changing its weight $w(e)$ to $\theta w(e)$, where $0<\theta \leq 1/2$ is a small scalar, and compute the Kirchhoff index of the resulting graph. The centrality of edge $e$ is reflected in the Kirchhoff index of the new graph: the larger the Kirchhoff index is, the more important the edge $e$ is. We define two equivalent metrics for edge centrality. One is the
Kirchhoff index of the new graph after edge  deactivation, the other is the difference of the Kirchhoff indices between the new graph and the original graph. For either edge centrality measure, we  give a fast algorithm to compute the $\eps$-approximation  for all the $m$ edges in nearly linear time. Furthermore, based on the Kirchhoff edge centrality index,   we propose a   vertex importance measure, with the centrality of a vertex being defined as  the Kirchhoff index of a new graph,
where all edges incident to it are deactivated,
and provide an efficient algorithm for estimating this new vertex centrality.


\subsection{Related Works}

Some edge centrality measures and related algorithms have been proposed.  Here we give a brief introduction to these metrics and their computational complexity. Moreover, we simply describe some work or techniques that partially  motivate this paper or relate to our algorithm.

Edge betweenness is probably  the most popular and most studied measure of edge importance. It measures the probablility that a shortest path between two vertices passes through a given edge. A fast algorithm for exact  computation of edge betweenness was developed by Brandes~\cite{Br01}.  For a graph with $n$ vertices and $m$ edge, the complexity for this efficient technique is $O(n m)$ and $O(nm+n^{2}\log n)$ for unweighted graphs and weighted graphs, respectively. In order to speed up the computation,  some approximate algorithms  have been proposed~\cite{BaKiMaMi07,BrPi07,GeSaSc08}.  All these approximate approaches aim at reducing the computation of shortest paths in different ways, without providing approximation guarantees.

Another edge importance measure is spanning edge centrality  first introduced in~\cite{TeMoCaRaFa13}.  The spanning edge centrality of an edge is equal to probability that the edge is used in a randomly selected spanning tree.  The best known exact algorithm has a running time $O(mn^{3/2})$. In order to compute  spanning edge centrality for massive networks, two fast
approximation algorithms~\cite{MaGaKoTe15, HaAkYo16} have been designed,  both having   theoretical guarantees
on their accuracy.

A third  measure for edge importance  is  current-flow centrality introduced by Brandes and Fleischer~\cite{BrFl05}. An edge has relatively significant importance, if it participates in many short paths
connecting pairs of vertices.  Brandes and Fleischer~\cite{BrFl05} provided an algorithm  with time complexity  $O(mn^{3/2}\log n)$, which can actually be dropped to $O(mn\log n)$ as shown in~\cite{MaGaKoTe15}.

Both  spanning edge centrality  and  current-flow centrality are closely related to effective resistance~\cite{MaGaKoTe15}.  In fact, the Kirchhoff edge centrality we propose belongs to the same class of \emph{electrica}l centrality measures. Moreover, this is the first definition of a global notion of centrality that can be provably approximated in nearly-linear time, which means that resistance based edge centrality for graphs may actually be easier to compute than other centrality measures  based on discrete structures, e.g. triangles or shortest paths.   

As many previous theoretical studies~\cite{BhHeNaTs15,KaPr17}, our work is also motivated by graph mining applications. In~\cite{BhHeNaTs15}, a streaming algorithm was developed for analyzing large-scale rapidly-changing graphs, which  maintains densest subgraphs in one pass and achieves time and space efficiency, whereas in our case,  effective  resistances are persevered under updates.  All the notions (dense subgraph~\cite{BhHeNaTs15}, triangle~\cite{KaPr17}, and  Kirchhoff  centrality) studied before or in the present paper have vast applications in network analysis, and their related computational challenges  fall within the scope of computer theory.

Our algorithms, in particular the resistance maintenance routines,
closely buid upon the sketching based inverse maintenance routine
from~\cite{LeSiWo15} and the computation of multiple partial
states of Gaussian eliminations from~\cite{DKP+17}.
The former maintains the inverse of a matrix under updates,
and is a critical routine for many graph
algorithms~\cite{San04,LS15, HaXu16}.
However, this often leads to dense matrices,
and we combine it with techniques from
graph sparsification~\cite{SS11,AbKuKoKrPe16,KyPaPeSa17,LeSu17}
to obtain our nearly-linear running times.
The additional need to maintain dot-products against
arbitrary vectors also leads us to incorporate iterative methods
in our routine for approximating the vertex Kirchoff centrality.
This demonstrates the robustness of our algorithm in combining
two different ways (Johnson-Lindenstrauss lemma and Schur complements)
of computing effective resistances.
We believe our results can be extended to provide more access to even
more graph quantities motivated by practical problems on graphs.

\subsection{Our Results}

For a graph $G$, we write $G\bsk e$ to denote the graph obtained from $G$
by deactivating edge $e$, i.e., decreasing the weight of  $e$  from $w(e)$ to $\theta w(e)$ for
some small $0 < \theta \leq 1/2$.
Let $\LL $ be  the Laplacian matrix of $G$, and let $\LL\bsk e$ denote the
Laplacian matrix of $G\bsk e$.
Then we can define two metrics for $\theta$-Kirchhoff centrality of an edge $e$,
denoted by $\mathcal{C}_\theta(e)$ and $\mathcal{C}^{\Delta}_\theta(e)$,
respectively.
$\mathcal{C}_\theta(e)$ is the Kirchhoff index of the graph $G\bsk e$, i.e.,
the sum of effective resistances over all vertex pairs in $G\bsk e$,
while  $\mathcal{C}^{\Delta}_\theta(e)$ is the difference between
the Kirchhoff indices of graph $G\bsk e$ and graph $G$.
Let $\Kf{G}$ denote the Kirchhoff index of graph $G$, then we have $\mathcal{C}_\theta(e) = \Kf{G\bsk e}$
and $\mathcal{C}^\Delta_\theta(e) = \Kf{G\bsk e} - \Kf{G}$.


As has been shown in~\cite{ElSpVaJaKo11}, the Kirchhoff index of a graph equals $n$ times $\trace{\LL^\dag}$,
where $\LL^\dag$ is the pseudoinverse of the graph's Laplacian matrix.
Thus, we have $\mathcal{C}_\theta(e) = n\trace{\Lke^\dag}$
and $\mathcal{C}^{\Delta}_\theta(e) = n\trace{\Lke^\dag} - n \trace{\LL^\dag}$.
To compute the exact value of $\theta$-Kirchhoff centrality for each edge,
a naive algorithm
would invert the matrix $\LL\bsk e$ for all $e\in E$.
Since a single inversion takes $O(n^\omega)$ time, where $\omega \approx 2.373$ is the matrix multiplication constant~\cite{wi12},
the naive algorithm runs in $O(n^\omega m)$ time for all the $m$ edges, which makes it untractable for large networks.

In this paper, we consider the scenario in which only approximate
values of $\theta$-Kirchhoff centrality are needed.
Such approximations are acceptable in many cases because
we only need to estimate relative importance of edges.
We give a randomized algorithm $\ECComp$ that computes
$\eps$-approximate Kirchhoff edge centrality $\mathcal{C}_\theta(e)$
for all the $m$ edges in $\tilde{O}(m\eps^{-4})$ time,
and a randomized algorithm $\ECComptwo$ that computes
$\eps$-approximate Kirchhoff edge centrality $\mathcal{C}^{\Delta}_\theta(e)$
for all the $m$ edges in $\tilde{O}(m \theta^{-2} \epsilon^{-2})$ time.
The key ingredients of algorithm $\ECComp$ are
Schur complements and Cholesky factorizations, 
which have been used in
various applications, such as solving linear systems in Laplacians~\cite{KLP+16,KS16}
and counting and sampling spanning trees~\cite{DKP+17,DPPR17}.
And the key technique for algorithm $\ECComptwo$ is the combination
of sketching with the Sherman-Morrison formula~\cite{ShMo50} from
efficient maintenances of matrix inverses for optimization~\cite{LeSiWo15}.

The performance of the algorithm $\ECComp$ is characterized in the following theorem.


\begin{restatable}[]{theorem}{thmedgeCentComp}
\label{lem:edgeCentComp}
	Given a connected undirected graph $G = (V,E)$ with $n$ vertices, $m$ edges, positive edge weights
  	$w : E \to \rea_{+}$, 
  	and scalars $0 < \theta \leq 1/2$, $0<\eps\leq1/2$,
  	the algorithm $\ECComp(G = (V,E), w, \theta,\epsilon)$ returns a set of pairs
  	$\hat{C} = \{ (e,\hat{c}_e) \mid e \in E \}$.
  	With high probability, the following statement holds: For $\forall e \in E$,
	\begin{align}
		\mathcal{C}_\theta(e) \approx_{\eps} \hat{c}_e,
	\end{align}
	where
	\[
		\mathcal{C}_\theta(e) = \sum\limits_{u,v \in V} \mathcal{R}_\mathrm{eff}^{G\bsk e}(u,v)
	\]
 is the sum of effective resistances $\mathcal{R}_\mathrm{eff}^{G\bsk e}(u,v)$ over all vertex pairs
 $u$ and $v$ in  graph  $G\bsk e$.
	The total running time of this algorithm is bounded by
	$O(m \eps^{-4} \log^2 m \log^7 n \operatorname{polyloglog}(n))$.
\end{restatable}

The proof of this theorem appears in Section~\ref{sec:main}.

In Theorem~\ref{lem:edgeCentComp}, $\theta$ is arbitrary and can even depend on $n$, e.g. $1/n$.
When $\theta$ is constant,
we can give a simpler algorithm $\ECComptwo$ that approximates
$\mathcal{C}_\theta^\Delta$-Kirchhoff edge centrality for all $m$ edges in $\tilde{O}(m\theta^{-2}\eps^{-2})$ time.
The idea is to use the Sherman-Morrison formula, which gives a fractional expression of the difference
between $\LL^\dag$ and $\kh{\LL\bsk e}^\dag$, where we can approximate the numerator by
the Johnson-Lindenstrauss lemma, and the denominator by estimating effective resistances.
The technique is similar to the approach in~\cite{LeSiWo15}.

\begin{restatable}[]{theorem}{thmedgebyjl}
\label{lem:edgebyjl}
	Given a connected undirected graph $G = (V,E)$ with $n$ vertices, $m$ edges, positive edge weights
  	$w : E \to \rea_{+}$, 
  	and scalars $0 < \theta \leq 1/2$, $0<\eps\leq1/2$,
  	the algorithm $\ECComptwo(G = (V,E), w, \theta,\epsilon)$ returns a set of pairs
  	$\hat{C} = \{ (e,\hat{c}^\Delta_e) \mid e \in E \}$.
  	With high probability, the following statement holds: For $\forall e \in E$,
	\begin{align}
		\mathcal{C}_\theta^\Delta(e) \approx_{\eps} \hat{c}^\Delta_e,
	\end{align}
	where
	\[
		\mathcal{C}_\theta^\Delta(e) = \Kf{G\bsk e} - \Kf{G}.
	\]
	The total running time of this algorithm is bounded by
	$O(m \theta^{-2} \eps^{-2} \log^{2.5} n \log(1/\eps) \operatorname{polyloglog}(n))$.
\end{restatable}

The proof of this theorem appears in Section~\ref{sec:edgebyjl}.
Its advantage is that for moderate values of $\theta$,
it can obtain a more accurate estimate of $\mathcal{C}_\theta^\Delta(e)$
even if $\Kf{G}$
is large.
However, when $\theta$ is small, the removal of high effective resistance
edges can cause a large error in this routine, and we are not guaranteed
to even get a good estimate of $\mathcal{C}_\theta(e)$ by adding this
result to an estimate of $\Kf{G}$.
As a result we believe both of our algorithms for estimating Kirchhoff edge
 centrality are of interest, and complement each other.

Based on the same idea of the definition for $\mathcal{C}_\theta^\Delta(e)$,
we can define a centrality measure for any vertex $v$,
which is the difference of Kirchhoff indices between the new
graph $G\bsk \Ev{v}$ and original graph $G$,
where $G\bsk \Ev{v}$ is obtained from $G$ by multiplying
the weights of all edges incident with $v$ by $\theta$.
We write $\mathcal{C}_\theta^\Delta(v)$ to denote the $\theta$-Kirchhoff vertex centrality of $v$.
In this situation, the matrix perturbation caused by removing the
neighborhood of $v$ is no longer rank 1, and we need to leverage
the approximate Schur complement routines from Section~\ref{sec:approx}
to compute these intermediate matrices.
Specifically, for constant $\theta$, we give an algorithm $\VCComp$ that
approximates $\theta$-Kirchhoff vertex centrality
for all $n$ vertices in $\tilde{O}(m\theta^{-2.5}\eps^{-4})$ time.

\begin{restatable}[]{theorem}{thmvertexbyjlschur}
\label{lem:vertexbyjlschur}
	Given a connected undirected graph $G = (V,E)$ with $n$ vertices, $m$ edges, positive edge weights
  	$w : E \to \rea_{+}$, 
  	and scalars $0 < \theta \leq 1/2$, $0<\eps\leq1/2$,
  	the algorithm $\VCComp(G = (V,E), w, \theta,\epsilon)$ returns a set of pairs
  	$\hat{C} = \{ (v,\hat{c}^\Delta_v) \mid v \in V \}$.
  	With high probability, the following statement holds: For $\forall v \in V$,
	\begin{align}
		\mathcal{C}_\theta^\Delta(v) \approx_{\eps} \hat{c}^\Delta_v,
	\end{align}
	where
	\[
		\mathcal{C}_\theta^\Delta(v) = \Kf{G\bsk \Ev{v}} - \Kf{G}
	\]
	and $\Ev{v} = \setof{(u,v) \,|\, u\sim v}$ is the set of edges incident with $v$.
	The total running time of this algorithm is bounded by
	$O(m(\theta^{-2}\eps^{-4}\log^9 n + \theta^{-2.5}\eps^{-4}\log^6 n \log(1/\eps)) \operatorname{polyloglog}(n))$.
\end{restatable}

The proof of this theorem appears in Section~\ref{sec:vertexbyjlschur}.

}


\begin{figure}
	\centering
   \includegraphics[width=0.3\textwidth]{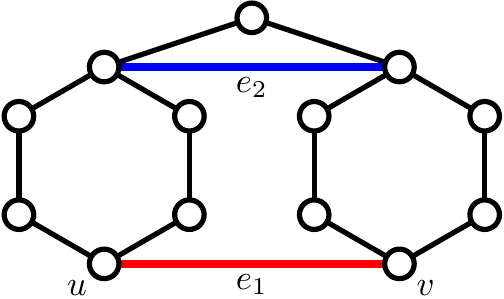}
   \caption{Betweenness cannot distinguish between edges $e_1$ and $e_2$.}
   \label{fig:btw}
\end{figure}

\subsection{Comparison With Other Measures}

In addition to the low computational complexity, the  $\theta$-Kirchhoff centrality is more discriminating than other edge centrality measures, such as edge  betweenness centrality and spanning edge centrality.  For example, in the graph in Figure~\ref{fig:btw},
the importance of edge $e_1$ and edge $e_2$ are different, which can be seen  by intuition. In fact, we can also understand this difference from the influences  when the two edges are deleted.  If   $e_1$ is removed,  the length of shortest path between vertices $u$ and $v$ increases by 6, while  the removal of $e_2$  will increase  the length of shortest path between any pair of vertices by at most 1. However, the betweenness centrality  for   $e_1$ and   $e_2$ are the same,  being equal to  $18$, implying that betweenness centrality cannot differentiate  $e_1$ between  $e_2$. However, these two edges can be  discriminated by the $\theta$-Kirchhoff edge centrality.  Exact  computation shows that the $0.1$-Kirchhoff edge centrality for   $e_1$ and   $e_2$  is  $\mathcal{C}_{0.1}(e_1) = 132.65$ and $\mathcal{C}_{0.1}(e_2) = 112.34$, respectively. Thus, $e_1$ is relatively more important than $e_2$,  which agrees with our human intuition.

\begin{figure}
	\centering
  \includegraphics[width=0.4\textwidth]{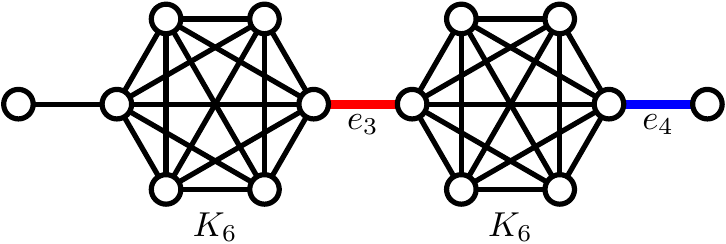}
   \caption{Spanning edge centrality cannot distinguish between edges $e_3$ and $e_4$.}
   \label{fig:spanning}
\end{figure}

We continue to show that the $\theta$-Kirchhoff centrality is also more discriminating than the spanning edge centrality.
By intuition, the importance for the two edges $e_3$ and  $e_4$ of the graph illustrated in Figure~\ref{fig:spanning} are
distinct. Unfortunately,   spanning edge centrality cannot distinguish  these two edges, since their  spanning edge centrality is identical, both equalling $1$.  In contrast, the $0.1$-Kirchhoff edge centrality of these two edges  $e_3$ and  $e_4$ is $\mathcal{C}_{0.1}(e_3) = 467.33$ and $\mathcal{C}_{0.1}(e_4) = 197.33$, respectively.  This implies that  $e_3$ plays a relatively more significant role than  $e_4$, which is consistent with our intuition.

To further show the capability of our $\theta$-Kirchhoff edge centrality to discriminate between different edges, we experimentally compare our measure $\mathcal{C}_\theta^\Delta$ with other metrics, including edge centrality, spanning edge centrality, and current-flow edge centrality. For each measure, we numerically evaluate the importance of each edge for some classic real-world networks\footnote{All data can be found at http://www-personal.umich.edu/\textasciitilde mejn/netdata/} in Table~1. The data sets are from published data mining related papers~\cite{Zac77,Knu93,New06,LSB+03,WS98}. Based on which we then compute the relative standard deviation for each centrality measure (as the authors did in~\cite{BWLM16}), where the relative standard deviation is defined as the standard deviation divided by the average.  Figure~\ref{fig:deviation} shows the relative standard deviation for all the centrality measures. It is always significantly higher for $\theta$-Kirchhoff edge centrality than it is for other measures, meaning that our measure has a better capability to distinguish between different edges.

\begin{table}\label{tabb}
	\centering
	\caption{Some classic real networks.}
	\begin{tabular}{|c|c|c|}
		\hline
		Network name & Number of vertices & Number of edges \\
		\hline
		Karate~\cite{Zac77} & 34 & 78 \\
		\hline
		Lesmis~\cite{Knu93} & 77 & 254 \\
		\hline
		Adjnoun~\cite{New06} & 112 & 425 \\
		\hline
		Dolphins~\cite{LSB+03} & 62 & 159 \\
		\hline
		Celegansneural~\cite{WS98} & 297 & 2148 \\
		\hline
	\end{tabular}
\end{table}

\begin{figure}
    \centering
  \includegraphics[width=0.8\textwidth]{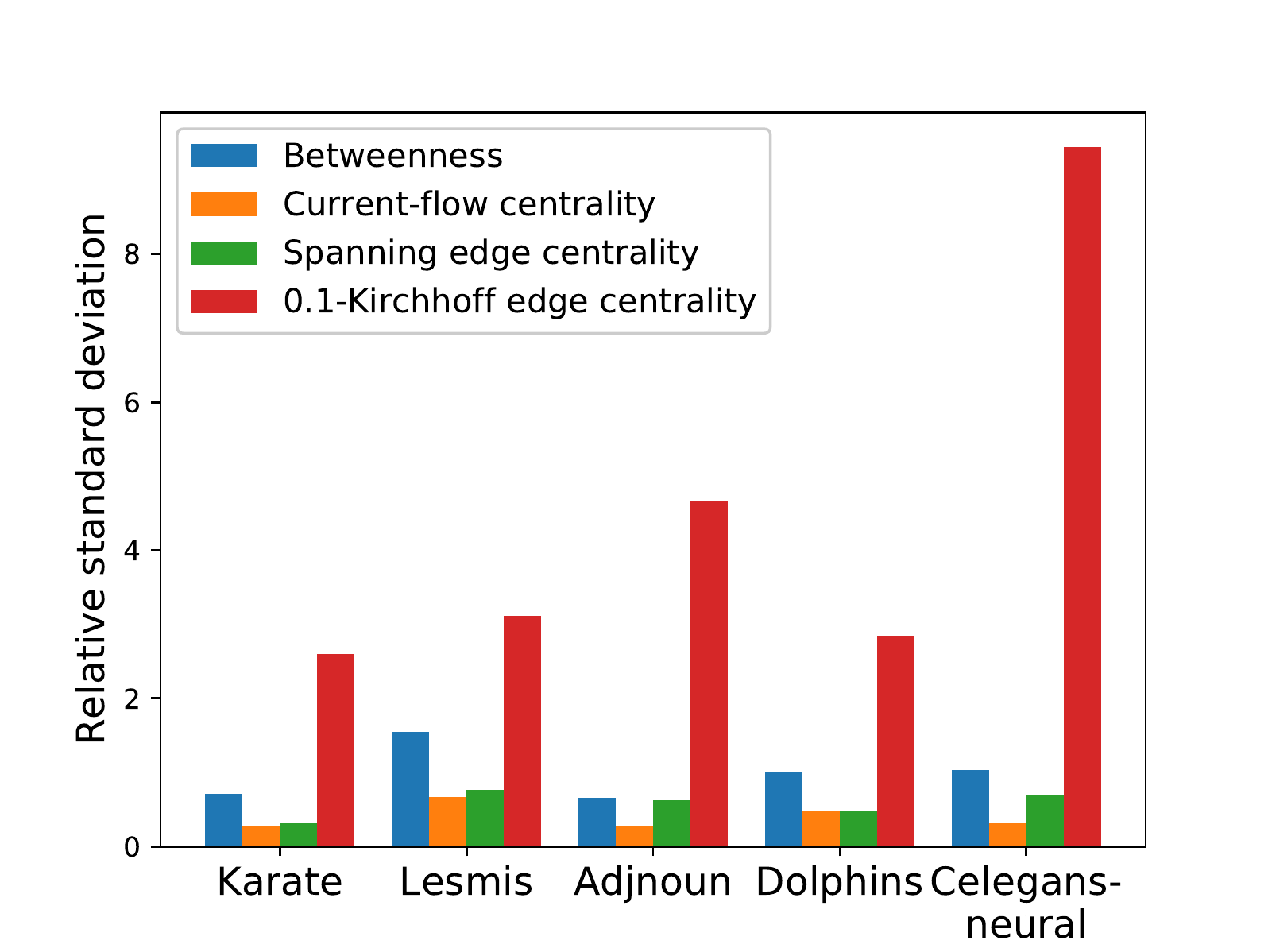}
  \caption{Relative standard deviation for different edge centrality measures.}
	\label{fig:deviation}
\end{figure}

\subsection{Organization}

The remaining part of the paper is organized as follows.
In Section~\ref{sec:prob}, we present the background and formulate the problem of computing $\theta$-Kirchhoff centrality.
In Section~\ref{sec:schur}, we introduce Schur complements and partial Cholesky factorizations,
and a lemma with regard to the performance of the approximate partial Cholesky factorization algorithm in~\cite{DKP+17}.
In Section~\ref{sec:main}, we introduce our algorithm $\ECComp$ that approximates $\mathcal{C}_\theta(e)$.
In Section~\ref{sec:edgebyjl}, we introduce our algorithm $\ECComptwo$ that approximates $\mathcal{C}_\theta^\Delta(e)$.
In Section~\ref{sec:vertexbyjlschur}, we introduce our algorithm $\VCComp$ that approximates $\mathcal{C}_\theta^\Delta(v)$.
In Section~\ref{sec:conclude}, we give our conclusion and discuss
some directions for future works.

\section{Background and the Problem}\label{sec:prob}

\subsection{Multiplicative Approximation of Scalars and Matrices}

We use the notion of $\eps$-approximation in~\cite{PS14}.

Let $a,b\geq 0$ be two nonnegative scalars.
We say $a$ is an $\eps$-approximation of $b$ if
\begin{align}\label{eq:epsApprox}
	\exp(-\eps)\, a \leq b \leq \exp(\eps)\, a.
\end{align}
We write $a \approx_{\eps} b$ to denote Eq.~(\ref{eq:epsApprox}).

For two matrices $\AA$ and $\BB$, we write $\AA \preceq \BB$ to indicate that $\BB - \AA$ is positive semidefinite.
We say $\AA$ is an $\eps$-spectral approximation of $\BB$ if
\begin{align}\label{eq:epsSpecApprox}
	\exp(-\eps)\, \AA \preceq \BB \preceq \exp(\eps)\, \AA.
\end{align}
We write $\AA \approx_{\eps} \BB$ to denote Eq.~(\ref{eq:epsSpecApprox}).

Note that these two relations are symmetric. Namely, $a \approx_{\eps} b$ implies $b \approx_{\eps} a$ and
$\AA \approx_{\eps} \BB$ implies $\BB \approx_{\eps} \AA$.

The following facts are basic properties of $\eps$-approximation:

\begin{fact}
\label{fact:Approximations}
	For nonnegative scalars $a,b,c,d \geq 0$, positive semidefinite matrices $\AA,\BB,\CC,\DD$,
	\begin{enumerate}
		\item \label{apxstart} if $a \approx_\eps b$, then $a + c \approx_\eps b + c$;
		\item if $a \approx_\eps b$ and $c \approx_\eps d$, then $a + c \approx_\eps b + d$;
		\item if $a \approx_{\eps_1} b$ and $b \approx_{\eps_2} c$, then $a \approx_{\eps_1 + \eps_2} c$;
		\item if $a$ and $b$ are positive such that $a \approx_\eps b$, then $1/a \approx_\eps 1/b$;
		\item \label{apxnum} if $a \approx_\eps b$, then $ac \approx_\eps bc$;
		\item \label{apxmat} if $\AA \approx_\eps \BB$, then $\AA + \CC \approx_\eps \BB + \CC$;
		\item if $\AA \approx_\eps \BB$ and $\CC \approx_\eps \DD$, then $\AA + \CC \approx_\eps \BB + \DD$;
		\item if $\AA \approx_{\eps_1} \BB$ and $\BB \approx_{\eps_2} \CC$, then $\AA \approx_{\eps_1 + \eps_2} \CC$;
		\item if $\AA$ and $\BB$ are positive definite matrices such that $\AA \approx_\eps \BB$,
		then $\AA^{-1} \approx_\eps \BB^{-1}$;
		\item 
\label{part:CompositionMatrix}
if $\AA \approx_\eps \BB$ and $\VV$ is a matrix, then $\VV^\top \AA \VV \approx_\eps \VV^\top \BB \VV$.
	\end{enumerate}
\end{fact}

\subsection{Graphs and Laplacians}

We consider a connected undirected graph $G = (V,E)$ with $n$ vertices, $m$ edges, and positive edge weights
$w : E \to \rea_{+}$. For a pair of vertices $u,v\in E$, we write $u\sim v$ to denote $(u,v) \in E$.
The Laplacian matrix of $G$ is an $n\times n$ matrix $\LL$ with
the entry on its $u^{\mathrm{th}}$ row and $v^{\mathrm{th}}$ column being
\begin{align*}
	\LL(u,v) = \begin{cases}
		- w(u,v)\quad & \mathrm{if}\ u\sim v, \\
		\mathrm{deg}(u)\quad & \mathrm{if}\ u = v, \\
		0\quad & \mathrm{otherwise},
	\end{cases}
\end{align*}
where $\mathrm{deg}(u) = \sum\limits_{u\sim v} w(u,v)$.
If $A$ and $B$ are two sets of vertices in $G$,
we write $\LL_{AB}$ to denote the submatrix of $\LL$
with rows corresponding to $A$ and columns corresponding to $B$.

Let $\ee_i$ denote the $i^{\mathrm{th}}$ standard basis vector, and $\bb_{u,v} = \ee_u - \ee_v$.
We fix an arbitrary orientation of the edges in $G$.
For each edge $e\in E$, we define $\bb_e = \bb_{u,v}$, where $u$ and $v$
are head and tail of $e$, respectively. It is easy to show that $\LL = \sum\nolimits_{e\in E} w(e)\bb_e \bb_e^\top$.
We refer to $w(e)\bb_e \bb_e^\top$ as the Laplacin of $e$.

It is immediate that $\LL$ is positive semidefinite since
\begin{align*}
	\xx^\top \LL \xx = \xx^\top \kh{ \sum\limits_{e\in E} w(e)\bb_e \bb_e^\top } \xx\
	= \sum\limits_{e\in E} w(e)\xx^\top \bb_e \bb_e^\top \xx
	= \sum\limits_{e\in E} w(e) \kh{ \xx^\top \bb_e }^2 \geq 0
\end{align*}
holds for any $\xx \in \mathbb{R}^n$.

\subsection{The Pseudoinverse, Effective Resistances, and Kirchhoff Index}

Since $\LL$ is positive semidefinite, we can diagonalize it and write
\begin{align*}
	\LL = \sum\limits_{i=1}^{n-1} \lambda_i \vv_i \vv_i^\top,
\end{align*}
where $\lambda_1,\ldots,\lambda_{n-1}$ are the nonzero eigenvalues of $\LL$ and
$\vv_1,\ldots,\vv_{n-1}$ are the corresponding orthonormal eigenvectors.
The pseudoinverse of $\LL$ is defined as
\begin{align*}
	\LL^\dag = \sum\limits_{i=1}^{n-1} \frac{1}{\lambda_i} \vv_i \vv_i^\top.
\end{align*}

It is not hard to show that
if $\LL$ and $\LH$ are Laplacian matrices of connected graphs and $\LL \approx_\eps \LH$, then
$\LL^\dag \approx_\eps \LH^\dag$.

We then give the definitions of effective resistance and Kirchhoff index:

\begin{definition}[Effective Resistance]
For a connected undirected graph $G = (V,E)$,	 the effective resistance between $u$ and $v$ is defined as
	\begin{align*}
		\mathcal{R}_{\mathrm{eff}}^G(u,v) = \bb_{u,v}^\top \LL^\dag \bb_{u,v}.
	\end{align*}
\end{definition}

\begin{definition}[Kirchhoff Index]
	The Kirchhoff index $\Kf{G}$ of a graph $G = (V,E)$ is defined as the sum of effective resistances over all vertex pairs.
	Namely,
	\begin{align*}
		\Kf{G} = \sum\limits_{u,v\in V} \mathcal{R}_{\mathrm{eff}}^G(u,v).
	\end{align*}
\end{definition}
For a graph, Kirchhoff index is a measure of its overall connectedness. A graph with smaller Kirchhoff index is better connected on an average. It is known that the Kirchhoff index of a graph equals $n$ times
the sum of reciprocals of nonzero eigenvalues of $\LL$~\cite{ElSpVaJaKo11},
and hence also equals $n$ times the trace of $\LL^\dag$. We give this
relation in the following Fact:

\begin{fact}\label{fact:kirchtr}
	Let $\lambda_1,\ldots,\lambda_{n-1}$ be the nonzero eigenvalues of $\LL$.
	The Kirchhoff index of graph $G$ satisfies
	\begin{align*}
		\Kf{G} = n\sum\limits_{i=1}^{n-1} \frac{1}{\lambda_i} = n \trace{\LL^\dag}.
	\end{align*}
\end{fact}

By Rayleigh's Monotonicity Law~\cite{ElSpVaJaKo11},
the effective resistance between any pair of vertices can only increase when edges are deleted or edge weights are decreased.
Since the Kirchhoff index is the sum of effective resistances over all vertex pairs, we have the following Fact:

\begin{fact}\label{lem:kfmono}
	The Kirchhoff index of a graph does not decrease when edges are deleted or edge weights are decreased.
\end{fact}

\subsection{Kirchhoff Edge Centrality}


With Fact~\ref{lem:kfmono}, it is reasonable to measure the importance  of an edge $e$ in graph $G$ by 
the Kirchhoff index of the new graph in which $e$ is deactivated.
We formalize the notion of edge deactivation by defining $\theta$-deletion of an edge.


\begin{definition}[$\theta$-Deletion]
	Let $0< \theta \leq 1/2$ be a scalar.
	For an edge $e\in E$, the $\theta$-deletion of $e$ is to decrease its weight from $w(e)$ to $\theta w(e)$.
	If we use $G\bsk e$ to denote the graph obtained from $G$ by $\theta$-deleting edge $e$,
	and $\LL\bsk e$ to denote the Laplacin matrix corresponding to $G\bsk e$, we have
	\begin{align*}
		\LL\bsk e = \LL - (1 - \theta) w(e) \bb_e \bb_e^\top.
	\end{align*}
	We can also define $\theta$-Deletion of an edge set $B \subset E$, which is to
	decrease the weight of each edge $e\in B$ from $w(e)$ to $\theta w(e)$. Similar notations are
	$G\bsk B$ and
	\begin{align*}
		\LL\bsk B = \LL - (1 - \theta) \sum\limits_{e\in B} w(e)\bb_e \bb_e^\top.
	\end{align*}
\end{definition}

We then give the definitions of $\theta$-Kirchhoff edge centrality
$\mathcal{C}_\theta(e)$ and $\mathcal{C}_\theta^\Delta(e)$.
\begin{definition}[$\theta$-Kirchhoff Edge Centrality $\mathcal{C}_\theta$]
	Let $0< \theta \leq 1/2$ be a scalar.
	For an edge $e\in E$, its $\theta$-Kirchhoff edge centrality $\mathcal{C}_\theta(e)$ is defined as
	the Kirchhoff index of the graph obtained from $G$ by $\theta$-deleting $e$. Namely,
	\begin{align*}
		\mathcal{C}_\theta(e) = \Kf{G\bsk e}.
	\end{align*}
\end{definition}

\begin{definition}[$\theta$-Kirchhoff Edge Centrality $\mathcal{C}_\theta^\Delta$]
	Let $0< \theta \leq 1/2$ be a scalar.
	For an edge $e\in E$,
	its $\theta$-Kirchhoff edge centrality $\mathcal{C}_\theta^\Delta(e)$ is defined as
	the increase of the Kirchhoff index of the graph upon edge $e$'s $\theta$-deletion. Namely,
	\begin{align*}
		\mathcal{C}_\theta^\Delta(e) = \Kf{G\bsk e} - \Kf{G}.
	\end{align*}
\end{definition}

Clearly, these two definitions of Kirchhoff edge centrality lead to the same ranking of edges.

Following the above two definitions of Kirchhoff edge centrality, we can also define $\theta$-centrality of
a vertex.

\begin{definition}[$\theta$-Kirchhoff Vertex Centrality]
\label{def:VertexRemoval}
	Let $0 < \theta \leq 1/2$ be a scalar.
	For a vertex $v\in V$,
	its $\theta$-Kirchhoff vertex centrality $\mathcal{C}_\theta^\Delta(v)$ is defined as
	the increase of the Kirchhoff index of the graph upon $\theta$-deletion of its incident edges. Namely,
	\begin{align*}
		\mathcal{C}_\theta^\Delta(v) = \Kf{G\bsk \Ev{v}} - \Kf{G},
	\end{align*}
	where $\Ev{v} = \setof{ (u,v)\,|\,u\sim v }$ is the set of edges incident with $v$.
\end{definition}


We now formulate the core problems of approximating $\theta$-Kirchhoff centrality:

\begin{problem}\label{prob}
	Given a connected undirected graph $G = (V,E)$ with $n$ vertices, $m$ edges, and positive edge weights
	$w : E \to \rea_{+}$, and scalars $0 < \theta \leq 1/2$, $0 < \eps \leq 1/2$, for each $e \in E$,
	find an $\eps$-approximation of
	its $\theta$-Kirchhoff edge centrality $\mathcal{C}_\theta(e)$.
\end{problem}

\begin{problem}\label{probedge}
	Given a connected undirected graph $G = (V,E)$ with $n$ vertices, $m$ edges, and positive edge weights
	$w : E \to \rea_{+}$, and scalars $0 < \theta \leq 1/2$, $0 < \eps \leq 1/2$, for each $e \in E$,
	find an $\eps$-approximation of
	its $\theta$-Kirchhoff edge centrality $\mathcal{C}_\theta^\Delta(e)$.
\end{problem}

\begin{problem}\label{probvert}
	Given a connected undirected graph $G = (V,E)$ with $n$ vertices, $m$ edges, and positive edge weights
	$w : E \to \rea_{+}$, and scalars $0 < \theta \leq 1/2$, $0 < \eps \leq 1/2$, for each $v \in V$,
	find an $\eps$-approximation of
	its $\theta$-Kirchhoff vertex centrality $\mathcal{C}_\theta^\Delta(v)$.
\end{problem}


\section{Schur Complements and Partial Cholesky Factorizations}\label{sec:schur}

In this section, we introduce Schur complements and partial Cholesky factorizations,
which are key techniques in our algorithm.

\subsection{Preliminaries}

We first give the definitions of Schur complements and partial Cholesky factorizations
according to~\cite{KS16,DKP+17}.

\begin{definition}[Schur Complement]
Suppose $\LL$ is the Laplacian of an undirected positive-weighted connected graph,
and $\LL(:,i)$ is the $i^{\text{th}}$ column of $\LL$. For a vertex $v_1$,
\[ \SS^{(1)} \defeq \LL - \frac{1}{\LL(v_1,v_1)}\LL(:,v_1) \LL(:,v_1)^{\top}\]
is called the \emph{Schur complement} of $\LL$ with respect to vertex
$v_1$.
The operation of subtracting $\frac{1}{\LL(v_1,v_1)}\LL(:,v_1) \LL(:,v_1)^{\top}$ from $\LL$
is called the \emph{elimination} of vertex $v_1$.
Suppose we perform a sequence of eliminations, where in the $i^{th}$
step, we select a vertex
$v_{i} \in V\setminus \setof{v_{1}, \ldots, v_{i-1}}$ and eliminate
the vertex $v_{i}.$ We define
\begin{align*}
  \alpha_{i} &= \SS^{(i-1)}(v_{i},v_{i}), \\
  \cc_{i} &= \frac{1}{\alpha_{i}} \SS^{(i-1)}(:,v_{i}), \\
  \SS^{(i)} &= \SS^{(i-1)} - \alpha_{i} \cc_{i}\cc_{i}^{\top}.
\end{align*}
Then $\SS^{(i)}$ is called the Schur complement with respect to vertices $\setof{v_{1}, \ldots, v_{i}}$.
Let $C = \setof{v_{1}, \ldots, v_{i}}$, then we also write $\SC(\LL, C) = \SS^{(i)}_{CC}$ to denote the Schur complement
of $\LL$ onto $C$.
\end{definition}

\begin{definition}[Partial Cholesky factorization]
Suppose we eliminate a sequence of vertices $v_{1}, \ldots, v_{i}$.
Let $\matlow$ be the
$n \times i$ matrix
with $\cc_{j}$ as its $j^{\textrm{th}}$ column,
and $\calD$ be the $i \times i$ diagonal matrix
$\calD(j,j) = \alpha_{j}$, then
\[
\LL = \SS^{(i)} + \sum_{j=1}^{i} \alpha_{j} \cc_{j} \cc_{j}^{\top}
=\SS^{(i)} + \matlow \calD \matlow^{\top}.
\]
Let us write $F = \setof{v_{1},\ldots,v_{i}}$, and $C = V\setminus F$.
Let $\SS$ be the submatrix of $\SS^{(i)}$ with rows and columns corresponding to vertices in $C$, i.e., $\SS = \SS^{(i)}_{CC}$.
Since $\SS^{(i)}_{CC}$ contains all nonzero entries of $\SS^{(i)}$,
we can write
$\matlow = \begin{pmatrix} \matlow\sm{F}{F} \\ \matlow\sm{C}{F} \end{pmatrix} $ and
\begin{align}
	\LL
	=
 	\begin{pmatrix}
    	\matlow\sm{F}{F} \\
    	\matlow\sm{C}{F}
  	\end{pmatrix}
  	\calD
  	\begin{pmatrix}
    	\matlow\sm{F}{F} \\
    	\matlow\sm{C}{F}
  	\end{pmatrix}^\top
  	+
  	\begin{pmatrix}
  		\matzero_{FF} & \matzero_{FC} \\
  		\matzero_{CF} & \SS
  	\end{pmatrix}
  	=
	\begin{pmatrix}
   		\matlow\sm{F}{F} & \matzero \\
	    \matlow\sm{C}{F} & \II\sm{C}{C}
	\end{pmatrix}
  	\begin{pmatrix}
  		\calD & \matzero \\
  		\matzero & \SS
  	\end{pmatrix}
  	\begin{pmatrix}
   		\matlow\sm{F}{F} & \matzero \\
    	\matlow\sm{C}{F} & \II\sm{C}{C}
  	\end{pmatrix}^{\top}.
\label{eq:partialCholesky}
\end{align}
Here $\begin{pmatrix}
    \matlow_{FF} & \matzero \\
    \matlow_{CF} & \II_{CC}
  \end{pmatrix}$ is a lower triangular matrix up to row exchanges, and $\calD$ is diagonal. Eq.~(\ref{eq:partialCholesky}) is known as \emph{partial Cholesky factorization}.
\end{definition}

It is known that Schur complements of a Laplacian are also Laplacians:
\begin{fact}[Fact 5.1 of~\cite{DKP+17}]
	The Schur complement of a Laplacian w.r.t. vertices $v_{1}, \ldots, v_{i}$ is a Laplacian.
\end{fact}

\subsection{Commutativity With Edge Deletions}

According to~\cite{KS16},
we can write the Schur complement w.r.t. a vertex $v_1$ as
\begin{align}
	\SS^{(1)} 
	\label{eq:schurAdd}
	&= \sum_{e \in E : e \not\ni v_{1}} w(e) \bb_{e}
	\bb_{e}^{\top} + \sum_{u \sim v_1} \sum_{v \sim v_1}
	\frac{w(u,v_1) w(v,v_1)}
	{\mathrm{deg}(v_1)}
	\bb_{u,v} \bb_{u,v}^{\top},
\end{align}
where the first term on rhs is the Laplacian corresponding to the edges not incident with $v_1$,
and the second term on rhs is a Laplacian whose edges are supported on $V\setminus \setof{v_1}$.
Thus, $\SS^{(1)}$ can be seen as a multigraph obtained by adding edges to $G\left[V\setminus \setof{v_1}\right]$,
the induced graph of $G$ on $V \setminus \setof{v_1}$. By induction,
for all $i$, $\SS^{(i)}$ can be seen as a multigraph obtained by adding edges to
$G\left[V\setminus \setof{v_1, \ldots, v_i}\right]$, the induced graph of $G$ on $V \setminus \setof{v_1, \ldots, v_i}$.
Also, by Eq.~(\ref{eq:schurAdd}), edges added to $G\left[V\setminus \setof{v_1}\right]$ to obtain $\SS^{(1)}$
are fully determined by edges incident with	 $v_1$ in the original graph $G$. By induction, for all $i$,
edges added to $G\left[V\setminus \setof{v_1, \ldots, v_i}\right]$ to obtain $\SS^{(i)}$ are fully determined
by edges incident with $\setof{v_1, \ldots, v_i}$ in the original graph $G$.
Thus, deletions (or $\theta$-deletions) performed to edges with both endpoints in $V \setminus \setof{v_1, \ldots, v_i}$
commute with taking partial Cholesky factorization. Therefore, we have the following lemma:
\begin{lemma}\label{lem:choleskyDel}
	Given a connected undirected graph $G =(V,E)$, with positive edge weights $w : E \to \rea_{+}$,
	and associated Laplacian $\LL$,
	a set of vertices $F = \setof{v_1, \ldots, v_i} \in V$.
	Let $C = V\setminus F$, and the partial Cholesky factorization of $\LL$ be
  	\begin{align*}
  		\LL =
  		\begin{pmatrix}
    		\matlow\sm{F}{F} \\
    		\matlow\sm{C}{F}
	  	\end{pmatrix}
	  	\calD
	  	\begin{pmatrix}
	    	\matlow\sm{F}{F} \\
	    	\matlow\sm{C}{F}
	  	\end{pmatrix}^\top
	  	+
	  	\begin{pmatrix}
	  		\matzero_{FF} & \matzero_{FC} \\
	  		\matzero_{CF} & \SS
	  	\end{pmatrix}
	  	=
  		\begin{pmatrix}
	    	\matlow_{FF} & \matzero \\
	    	\matlow_{CF} & \II_{CC}
	  	\end{pmatrix}
	  	\begin{pmatrix}
  			\calD & \matzero \\
			\matzero & \SS
		\end{pmatrix}
  		\begin{pmatrix}
    		\matlow_{FF} & \matzero \\
    		\matlow_{CF} & \II_{CC}
  		\end{pmatrix}^{\top}.
	\end{align*}
  	For any edge $e$ whose endpoints are both in $C$, any $0\leq \theta < 1$,
    \begin{align*}
  		\LL\bsk e =
  		\begin{pmatrix}
    		\matlow\sm{F}{F} \\
    		\matlow\sm{C}{F}
	  	\end{pmatrix}
	  	\calD
	  	\begin{pmatrix}
	    	\matlow\sm{F}{F} \\
	    	\matlow\sm{C}{F}
	  	\end{pmatrix}^\top
	  	+
	  	\begin{pmatrix}
	  		\matzero_{FF} & \matzero_{FC} \\
	  		\matzero_{CF} & \SS \bsk e
	  	\end{pmatrix}
	  	=
  		\begin{pmatrix}
	    	\matlow_{FF} & \matzero \\
	    	\matlow_{CF} & \II_{CC}
	  	\end{pmatrix}
	  	\begin{pmatrix}
  			\calD & \matzero \\
			\matzero & \SS \bsk e
		\end{pmatrix}
  		\begin{pmatrix}
    		\matlow_{FF} & \matzero \\
    		\matlow_{CF} & \II_{CC}
  		\end{pmatrix}^{\top}.
	\end{align*}
\end{lemma}

\subsection{Approximate Partial Cholesky Factorization Algorithm}

In~\cite{DKP+17}, the authors give an approximate partial Cholesky factorization algorithm,
whose performance can be characterized in the following lemma:
\begin{lemma}[Lemma 5.7 of~\cite{DKP+17}, paraphrased] \label{lem:partialChol}
	There is an algorithm \\ $\partialChol(\LL,C,\eps)$ that
	when given a connected undirected graph $G =(V,E)$, with positive edge weights
	$w : E \to \rea_{+}$, and associated Laplacian $\LL$, a set of vertices $C \subset V$,
	and a scalar $0<\eps\leq1/2$, returns a decomposition $(\matlowtil,\calDtil,\SStil)$.
	With high probability, the following statement hold:
	\begin{align} \label{eq:lapErrorBounds}
  		\LL & \approx_{\eps} \LLtil,
	\end{align}
	where $F = V\setminus C$ and
	\begin{align}\label{eq:choleskyLLtil}
		\LLtil =
 		\begin{pmatrix}
	    	\matlowtil\sm{F}{F} \\
	    	\matlowtil\sm{C}{F}
	  	\end{pmatrix}
	  	\calDtil
	  	\begin{pmatrix}
	    	\matlowtil\sm{F}{F} \\
    		\matlowtil\sm{C}{F}
	  	\end{pmatrix}^\top
	  	+
	  	\begin{pmatrix}
	  		\matzero_{FF} & \matzero_{FC} \\
	  		\matzero_{CF} & \SStil
	  	\end{pmatrix}
		=
		\begin{pmatrix}
    		\matlowtil_{FF} & \matzero \\
    		\matlowtil_{CF} & \II_{CC}
  		\end{pmatrix}
  		\begin{pmatrix}
			\calDtil & \matzero \\
			\matzero & \SStil
		\end{pmatrix}
		\begin{pmatrix}
		    \matlowtil_{FF} & \matzero \\
		    \matlowtil_{CF} & \II_{CC}
		\end{pmatrix}^\top.
	\end{align}
	Here $\SStil$ is a Laplacian matrix whose edges are supported on $C$ such that $\SStil \approx_\eps \SC(\LL, C)$.
	Let $k = \sizeof{C} = n-\sizeof{F}$.
	The total number of non-zero entries in $\SStil$ is $O(k
	\eps^{-2}\log n)$.
	$\begin{pmatrix}
   		\matlowtil_{FF} & \matzero \\
	   	\matlowtil_{CF} & \II_{CC}
   	\end{pmatrix}$ is a lower triangular matrix up to row exchanges.
    The total number of non-zero entries in $\begin{pmatrix}
   		\matlowtil_{FF} & \matzero \\
	   	\matlowtil_{CF} & \II_{CC}
   	\end{pmatrix}$ is $O(m + n \eps^{-2}\log^3 n)$.  $\calDtil$ is a diagonal matrix.
	
		For any vector $\bb \in \mathbb{R}^n$,
		one can evaluate
		$\begin{pmatrix}
 			\matlowtil_{FF} & \matzero \\
   	 		\matlowtil_{CF} & \II_{CC}
  		\end{pmatrix}^{-1} \bb$ in $O(m + n \eps^{-2}\log^3 n)$ time.
  		For any vector $\cc \in \mathbb{R}^{\sizeof{F}}$,
		one can evaluate $\kh{\calDtil}^{-1} \cc$ in $O(\sizeof{F})$ time.
  	
	The total running time is bounded by
	$O((m\log^{3} n+ n \eps^{-2} \log^{5} n) \operatorname{polyloglog}(n))$.
\end{lemma}

	Comparing to \cite{DKP+17}, Lemma~\ref{lem:partialChol}
	\begin{enumerate}
		\item removes the failure probability factor $\delta$ and just claims high probability.
			The running time of the algorithm $\partialChol(\LL,C,\eps)$ in \cite{DKP+17} is \\
			$O((m\log n \log^{2}(n/\delta)+ n \eps^{-2} \log n \log^{4} (n/\delta)) \operatorname{polyloglog}(n))$.
			To gain high probability, we just set the failure probability $\delta$ to $1 / n^c$ for an arbitrary constant $c > 0$.
			Then we have the running time bounded by $O((m\log^{3} n+ n \eps^{-2} \log^{5} n) \operatorname{polyloglog}(n))$.
		\item emphasizes that inverses of matrices
			$\begin{pmatrix}
    			\matlowtil_{FF} & \matzero \\
		    	\matlowtil_{CF} & \II_{CC}
			\end{pmatrix}$
			and $\calDtil$ can both be applied quickly, as they can be treated as lower triangular matrix and diagonal matrix,
			respectively.
	\end{enumerate}

The following lemma shows that
edge additions performed within $C$ commute with taking approximate partial Choleksy factorization:

\begin{lemma}\label{lem:approxCholeksyAdd}
	Given a connected undirected multi-graph $G =(V,E)$, with positive edge weights
	$w : E \to \rea_{+}$, and associated Laplacian $\LL$, a set of vertices $C \subset V$,
	and an approximate partial factorization of $\LL$:
	\begin{align}\label{eq:approxCho}
		\LL \approx_\eps
		\begin{pmatrix}
    		\matlowtil_{FF} & \matzero \\
    		\matlowtil_{CF} & \II_{CC}
  		\end{pmatrix}
  		\begin{pmatrix}
			\calDtil & \matzero \\
			\matzero & \SStil
		\end{pmatrix}
		\begin{pmatrix}
		    \matlowtil_{FF} & \matzero \\
		    \matlowtil_{CF} & \II_{CC}
		\end{pmatrix}^\top,
	\end{align}
	where $F = V\setminus C$.	
	For any edge $e$ (not necessarily in $E$) with both endpoints in $C$ and a positive
	scalar $w_e > 0$,
	\begin{align}
		\LL + w_e \bb_e \bb_e^\top \approx_\eps
		\begin{pmatrix}
    		\matlowtil_{FF} & \matzero \\
    		\matlowtil_{CF} & \II_{CC}
  		\end{pmatrix}
  		\begin{pmatrix}
			\calDtil & \matzero \\
			\matzero & \SStil + \kh{ w_e \bb_e \bb_e^\top }_{CC}
		\end{pmatrix}
		\begin{pmatrix}
		    \matlowtil_{FF} & \matzero \\
		    \matlowtil_{CF} & \II_{CC}
		\end{pmatrix}^\top.
	\end{align}
\end{lemma}
\begin{proof}
	As multiplicative approximations are preserved under additions,
	by adding $w_e \bb_e \bb_e^\top$ to both sides of Eq.~(\ref{eq:approxCho}) we have
	\begin{align*}
		\LL + w_e \bb_e \bb_e^\top &\approx_\eps
		\begin{pmatrix}
    		\matlowtil_{FF} & \matzero \\
    		\matlowtil_{CF} & \II_{CC}
  		\end{pmatrix}
  		\begin{pmatrix}
			\calDtil & \matzero \\
			\matzero & \SStil
		\end{pmatrix}
		\begin{pmatrix}
		    \matlowtil_{FF} & \matzero \\
		    \matlowtil_{CF} & \II_{CC}
		\end{pmatrix}^\top + w_e \bb_e \bb_e^\top  \\
		&= \begin{pmatrix}
	    	\matlowtil\sm{F}{F} \\
	    	\matlowtil\sm{C}{F}
	  	\end{pmatrix}
	  	\calDtil
	  	\begin{pmatrix}
	    	\matlowtil\sm{F}{F} \\
    		\matlowtil\sm{C}{F}
	  	\end{pmatrix}^\top
	  	+
	  	\begin{pmatrix}
	  		\matzero_{FF} & \matzero_{FC} \\
	  		\matzero_{CF} & \SStil + \kh{ w_e \bb_e \bb_e^\top }_{CC}
	  	\end{pmatrix} \\
	  	&=\begin{pmatrix}
    		\matlowtil_{FF} & \matzero \\
    		\matlowtil_{CF} & \II_{CC}
  		\end{pmatrix}
  		\begin{pmatrix}
			\calDtil & \matzero \\
			\matzero & \SStil + \kh{ w_e \bb_e \bb_e^\top }_{CC}
		\end{pmatrix}
		\begin{pmatrix}
		    \matlowtil_{FF} & \matzero \\
		    \matlowtil_{CF} & \II_{CC}
		\end{pmatrix}^\top.	
	\end{align*}
\end{proof}

\section{Algorithm for Approximating $\theta$-Kirchhoff Edge Centrality $\mathcal{C}_\theta(e)$}\label{sec:main}

\subsection{Turning the Kirchhoff Index Into Quadratic Forms of $\vect{L}^\dag$}
\label{sec:index2norm}

By Fact~\ref{fact:kirchtr}, the Kirchhoff Index of a graph equals $n$ times the trace the Laplacian's pseudoinverse.
Although the explicit pseudoinverse of $\LL$ is hard to compute,
by taking approximate Cholesky factorizations \cite{KS16,DKP+17}, one
can approximate $\zz^\top \LL^\dag \zz$ for a $\zz \in \mathbb{R}^n$ quickly.
Thus, we can use Monte-Carlo methods to estimate trace of $\LL^\dag$.

The standard Monte-Carlo method for estimating the trace of an implicit matrix $\AA$ is due to
Hutchinson~\cite{Hut89}. The idea is to estimate the trace of $\AA$ by
$\frac{1}{M}\sum\nolimits_{i=1}^M \zz_i^\top \AA \zz_i$,
where the $\zz_i$'s are random $\pm 1$ vectors (i.e., independent Bernoulli entries).
Since there is $\expec{}{\zz_i^\top \AA \zz_i} = \trace{\AA}$, by the law of large numbers,
$\frac{1}{M}\sum\nolimits_{i=1}^M \zz_i^\top \AA \zz_i$ should be close to $\trace{\AA}$ when $M$ is large.
\cite{AT11} gives a rigorous bound on the number of Monte-Carlo samples required to achieve a maximum error $\eps$
with probability at least $1 - \delta$.
\begin{lemma}[Theorem 7.1 of~\cite{AT11}, paraphrased]
\label{lem:TraceEstimation}
	Let $\AA$ be a positive semidefinite matrix with rank $\mathrm{rank}(A)$.
	Let $\zz_1,\ldots,\zz_M$ be independent random $\pm 1$ vectors.
	Let $\eps, \delta$ be scalars such that $0 < \eps \leq 1/2$
	and $0 < \delta < 1$.
	For any $M \geq 24\eps^{-2} \ln(2 \mathrm{rank}(A) / \delta)$,
	the following statement holds with probability at least $1 - \delta$:
	\begin{align*}
		\frac{1}{M}\sum\limits_{i=1}^M \zz_i^\top \AA \zz_i \approx_\eps \trace{\AA}.
	\end{align*}
\end{lemma}

\begin{remark}
	We remark that the Hutchinson's method can be seen as Johnson-Lindenstrauss Lemma~\cite{JL84} in some sense.
	The reason is that since $\AA$ is positive semidefinite, one can write its trace as
	\begin{align*}
		\trace{\AA} = \trace{\AA^{1/2} \AA^{1/2}} = \len{\AA^{1/2}}_F^2,
	\end{align*}
	where $\len{\AA^{1/2}}_F^2$ can be seen as a sum of the squared lengths of the rows of $\AA^{1/2}$.
	By the discrete version of Johnson-Lindenstrauss Lemma from~\cite{Ach01},
	we can use a $n \times k$ random $\pm 1$ matrix $\QQ$, where $k = O(\eps^{-2} \log n)$, to reduce the dimensions:
	\begin{align*}
		\len{\AA^{1/2}}_F^2 \approx_{\eps} \frac{1}{k} \len{\AA^{1/2} \QQ}_F^2.
	\end{align*}
	This in turn implies
	\begin{align}\label{eq:jlhut}
		\trace{\AA} \approx_{\eps} \frac{1}{k} \sum\limits_{j=1}^k \qq_j^\top \AA \qq_j,
	\end{align}
	where $\qq_j$ is the $j^{\mathrm{th}}$ column of $\QQ$.
	The rhs of~(\ref{eq:jlhut}) can be seen as Hutchinson's method.
	Indeed,~\cite{AT11} used the discrete Johnson-Lindenstrauss Lemma from~\cite{Ach01}
	to prove their bound.
\end{remark}

Since $\LL^\dag$ is positive semidefinite and $\mathrm{rank}(\LL^\dag) = n - 1$,
by letting $\delta = 1/n$,
we have the following bound
on the number of Monte-Carlo samples required to achieve an $\eps$-approximation of $\trace{\LL^\dag}$ with high probability:
\begin{lemma}\label{lem:mcL}
	Let $\LL$ be a Laplacian matrix.
	Let $\zz_1,\ldots,\zz_M$ be independent random $\pm 1$ vectors.
	Let $\eps$ be a scalar such that $0 < \eps \leq 1/2$.
	For any $M \geq 48\eps^{-2} \ln(2n)$,
	the following statement holds with probability at least $1 - 1/n$:
	\begin{align*}
		\frac{1}{M}\sum\limits_{i=1}^M \zz_i^\top \LL^\dag \zz_i \approx_\eps \trace{\LL^\dag}.
	\end{align*}
\end{lemma}

A direct conclusion of Lemma~\ref{lem:mcL} is that for an edge $e\in E$, its $\theta$-Kirchhoff edge centrality
satisfies
\begin{align*}
	\mathcal{C}_\theta(e) = \mathcal{K}(G\bsk e) = n \trace{\kh{\LL\bsk e}^\dag} \approx_\epsilon
	\frac{n}{M}\sum\limits_{i=1}^M \zz_i^\top \Lke^\dag \zz_i.
\end{align*}
Thereby, the task of approximating the $\theta$-Kirchhoff edge centrality for all
$e\in E$ can be divided
into $O(\eps^{-2}\log n)$ independent tasks, each of which is to compute quadratic forms $\zz^\top \Lke^\dag \zz$ for
a fixed $\zz\in \mathbb{R}^n$ for all $e\in E$. We formulate these tasks in the following problem:

\begin{problem}\label{prob:norm}
	Given a connected undirected graph $G = (V,E)$ with $n$ vertices, $m$ edges, positive edge weights
	$w : E \to \rea_{+}$, and associated Laplacian $\LL$,
	a set of edges $\EQ\subset E$ such that every vertex in $V$ is incident to some edge $e \in \EQ$,
	a scalar $0 < \theta \leq 1/2$,
	and a vector $\zz \in \mathbb{R}^n$,
	find (approximately) $\zz^\top \LL^\dag \zz$ for all $e \in \EQ$.
\end{problem}


\subsection{Computing Quadratic Forms of $\vect{L}^\dag$
Upon Edge Deactivation}
\label{sec:phase2}

The idea of solving Problem~\ref{prob:norm} is to use recursions based on partial Cholesky factorizations.
We summarize the key steps in the following enumeration:

\begin{enumerate}
	\item If $\LL$ only have $O(1)$ vertices, invert $\LL\bsk e$ to compute $\zz^\top \Lke^{\dag} \zz$ for all $e\in \EQ$ and return.
	\item Divide edges in $\EQ$ into $E^{(1)}, E^{(2)}$ with equal sizes.
	\item \label{item:start} Let $C$ denote endpoints of edges in $E^{(1)}$ and $F = V\setminus C$.
	\item \label{item:Choleksy} By taking (approximate) partial Cholesky factorization of $\LL$,
		find a vector $\yy = \begin{pmatrix} \yy_F \\ \yy_C \end{pmatrix}$, a diagonal matrix $\DD_{\sizeof{F}\times \sizeof{F}}$,
		and a Laplacian matrix $\SS$ whose edges are supported on $C$,
		such that for each edge $e \in E^{(1)}$,
		$\zz^\top \Lke^{\dag} \zz$ can be evaluated by computing $\yy_F^\top \DD^{-1} \yy_F + \yy_C^\top \kh{\SS\bsk e}^\dag \yy_C$.
	\item \label{item:end} Compute $\yy_F^\top \DD^{-1} \yy_F$ by inverting $\DD$ and
		$\yy_C^\top \kh{\SS\bsk e}^\dag \yy_C$ for all $e\in E^{(1)}$ by recursion,
		then use $\yy_F^\top \DD^{-1} \yy_F + \yy_C^\top \kh{\SS\bsk e}^\dag \yy_C$ to evaluate
		$\zz^\top \Lke^{\dag} \zz$ for all $e \in E^{(1)}$.
	\item Repeat steps~\ref{item:start}\,-\,\ref{item:end} to $E^{(2)}$ to compute $\zz^\top \Lke^\dag \zz$
	for all $e \in E^{(2)}$. 
\end{enumerate}

The reason that in Step~\ref{item:end} we can compute $\yy_C^\top \kh{\SS\bsk e}^\dag \yy_C$ for each $e\in E^{(1)}$ by recursion
is that $\SS$ is a Laplacian matrix whose edges are supported on $C$,
and hence to compute $\yy_C^\top \kh{\SS\bsk e}^\dag \yy_C$ for all $e\in E^{(1)}$
is just a smaller-sized version of Problem~\ref{prob:norm} in which $\LL = \SS$ and $\EQ = E^{(1)}$.

In the rest of this subsection we give first an algorithm that solves Problem~\ref{prob:norm} exactly and then
an algorithm that solves Problem~\ref{prob:norm} approximately.

\subsubsection{Computing Exact Quadratic Forms of $\vect{L}^\dag$ Upon Edge Deactivation} \label{sec:exact}

We first give an algorithm $\exactQuad(\LL, \EQ, w, \zz, \theta)$ that computes the exact value of $\zz^\top \Lke^{\dag} \zz$
for a fixed $\zz \in \mathbb{R}^n$  for all $e\in \EQ$ (Here $w$ is the edge weight function).

In this algorithm, we find $\yy$, $\DD$, and $\SS$ in step~\ref{item:Choleksy}
by eliminating vertices in $F$ and
obtain an exact partial Cholesky factorization
of $\LL$. The following Lemma shows how to find them when an exact
partial Cholesky factorization of $\LL$ is given:

\begin{lemma}\label{lem:exactRecursion}
	For a graph $G = (V,E)$ with associate Laplacian $\LL$ and a set of vertices $C\subset V$.
	Let $F = V\setminus C$, and the partial Choleksy factorization of $\LL$ be
	\begin{align*}
   		\LL 
   		=
  		\begin{pmatrix}
	   		\matlow_{FF} & \matzero \\
	    	\matlow_{CF} & \II_{CC}
	  	\end{pmatrix}
	  	\begin{pmatrix}
  			\calD & \matzero \\
			\matzero & \SS
		\end{pmatrix}
  		\begin{pmatrix}
   			\matlow_{FF} & \matzero \\
   			\matlow_{CF} & \II_{CC}
		\end{pmatrix}^{\top}.\
	\end{align*}
	Let $\yy = \begin{pmatrix} \yy_F \\ \yy_C \end{pmatrix} =
 			\begin{pmatrix}
    			\matlow_{FF} & \matzero \\
    			\matlow_{CF} & \II_{CC}
  			\end{pmatrix}^{-1}
 			\zz$,
	then for each edge $e\in E$ with both endpoints in $C$
	the following statement holds:
	\begin{align*}
		\zz^\top \Lke^\dag \zz
		= \yy_F^\top \kh{\calD^{-1}} \yy_F + \yy_C^\top \kh{ \SS \bsk e }^\dag \yy_C.
	\end{align*}
\end{lemma}
\begin{proof}
	By Lemma~\ref{lem:choleskyDel}, $\theta$-deletions performed to edges with
	both endpoints in $C$ commute with taking
	partial Cholesky factorization. Thus, for each $e\in E$ with both endpoints in $C$, we have
		\begin{align}\label{eq:exact1}
			\LL\bsk e &=
			\begin{pmatrix}
	    		\matlow_{FF} & \matzero \\
	    		\matlow_{CF} & \II_{CC}
	  		\end{pmatrix}
	  		\begin{pmatrix}
  				\calD & \matzero \\
				\matzero & \SS \bsk e
			\end{pmatrix}
  			\begin{pmatrix}
    			\matlow_{FF} & \matzero \\
    			\matlow_{CF} & \II_{CC}
  			\end{pmatrix}^{\top}.
  		\end{align}
  		Inverting both sides of Eq.~(\ref{eq:exact1}) leads to
  		\begin{align*}
			\Lke^\dag &= \begin{pmatrix}
	    		\matlow_{FF} & \matzero \\
	    		\matlow_{CF} & \II_{CC}
	  		\end{pmatrix}^{-\top}
	  		\begin{pmatrix}
  				\calD^{-1} & \matzero \\
				\matzero & \kh{ \SS \bsk e }^\dag
			\end{pmatrix}
  			\begin{pmatrix}
    			\matlow_{FF} & \matzero \\
    			\matlow_{CF} & \II_{CC}
  			\end{pmatrix}^{-1}.
  		\end{align*}
  		Substituting
  		$
  			\yy = \begin{pmatrix} \yy_F \\ \yy_C \end{pmatrix} =
 			\begin{pmatrix}
    			\matlow_{FF} & \matzero \\
    			\matlow_{CF} & \II_{CC}
  			\end{pmatrix}^{-1}
 			\zz
  		$,
  		we obtain
  		\begin{align*}
			\zz^\top \Lke^\dag \zz
			&= \begin{pmatrix} \yy_F & \yy_C \end{pmatrix}
			\begin{pmatrix}
			\calD^{-1} & \matzero \\
			\matzero & \kh{ \SS \bsk e }^\dag
			\end{pmatrix}
			\begin{pmatrix} \yy_F \\ \yy_C \end{pmatrix} \\
			&= \yy_F^\top \kh{\calD^{-1}} \yy_F + \yy_C^\top \kh{ \SS \bsk e }^\dag \yy_C.
		\end{align*}
This completes the proof.
\end{proof}


We give the pseudocode for $\exactQuad$ in Algorithm~\ref{alg:exactNorm}.
Its performance is characterized in Lemma~\ref{lem:exactNorm}. \\

\begin{restatable}[]{lemma}{thmexactNorm}
	\label{lem:exactNorm}
	Given a connected undirected graph $G = (V,E)$ with $n$ vertices, $m$ edges, positive edge weights
	$w : E \to \rea_{+}$, and associated Laplacian $\LL$,
	a set of edges $\EQ\subset E$ such that every vertex in $V$ is incident with some edge $e \in \EQ$,
	a vector $\zz\in \mathbb{R}^n$,
  	and a scalar $0 < \theta \leq 1/2$,
  	the algorithm $\exactQuad(\LL, \EQ, w, \zz, \theta)$ returns a set of pairs
  	$N = \{ (e,n_e) \mid e \in \EQ \}$, where
	\[
	 n_e = \zz^\top \Lke^\dag \zz.
	\]
	The total running time of this algorithm is bounded by
	$O(n^{\omega - 1} m)$.
\end{restatable}

\begin{algorithm}[H]
	\caption{$\exactQuad(\LL, \EQ, w, \zz, \theta)$}
	\label{alg:exactNorm}
	\Input{
		$\LL$: A graph Laplacian.\\
		$\EQ$: A set of edges supported on vertices in $\LL$.\\
		$w$: An edge weight function.\\ 
		$\zz$: A vector whose dimension matches the number of vertices in $\LL$. \\
		$\theta$: The weight of  edge $e$  is temporarily changed from $w(e)$  to $\theta w(e)$ when 
			it is deactivated.
	}
	\Output{
		$N = \{(e,n_e) \mid e \in \EQ\}$: $n_e = \zz^\top \Lke^\dag \zz$.
	}
	
	Let $V$ denote the vertex set of $\LL$. \;
	\If{$|V| = 2$} {
		For every edge $e \in \EQ$, compute
		exact $n_e = \zz^\top \left(\LL - (1 - \theta) w(e)\bb_e\bb_e^\top\right)^\dag \zz$,
		then combine the results and return $N = \{(e,n_e) \mid e \in \EQ\}$.
		\label{line:exacto1}\;
	}
	Partition $\EQ$ into $E^{(1)}, E^{(2)}$ with $\sizeof{E^{(1)}} = \floor{\frac{\sizeof{\EQ}}{2}}$
		and $\sizeof{E^{(2)}} = \sizeof{\EQ} - \floor{\frac{\sizeof{\EQ}}{2}}$. \label{line:exactPartition}\;
	\For{$i = 1$ to $2$}{
		Let $C$ denote endpoints of edges in $E^{(i)}$ and $F = V \setminus{C}$. \;
		Eliminate all vertices in $F$ to get $\LL =
  			\begin{pmatrix}
	    		\matlow_{FF} & \matzero \\
		    	\matlow_{CF} & \II_{CC}
		  	\end{pmatrix}
		  	\begin{pmatrix}
  				\calD & \matzero \\
				\matzero & \SS
			\end{pmatrix}
  			\begin{pmatrix}
    			\matlow_{FF} & \matzero \\
    			\matlow_{CF} & \II_{CC}
	  		\end{pmatrix}^{\top}.
	    $ \label{line:takingCholesky} \;
		Compute $\yy = \begin{pmatrix} \yy_F \\ \yy_C \end{pmatrix} =\begin{pmatrix}
    			\matlow_{FF} & \matzero \\
    			\matlow_{CF} & \II_{CC}
  			\end{pmatrix}^{-1} \zz$. \label{line:mulStart}\;
		Compute $f = \yy_F^\top \left(\calD \right)^{-1} \yy_F$. \label{line:mulEnd} \;
		Call $\exactQuad(\SS, E^{(i)}, \yy_C, \theta)$
		 to compute $n_e^{(i)} = \yy_C^\top \kh{\SS\bsk e}^\dag \yy_C$ for all $e\in E^{(i)}$
			and store $(e, f + n_e^{(i)})$ in $N^{(i)}$.\;
	}
	\KwRet{$N = N^{(1)} \cup N^{(2)}$} \label{line:exactRet}
	
\end{algorithm}

\begin{proof}[Proof of Lemma~\ref{lem:exactNorm}]

As correctness is clear by Lemma~\ref{lem:exactRecursion},
we only need to prove the bound of running time.
Let $T(m)$ denote the running time of $\exactQuad(\LL, \EQ, w, \zz, \theta)$, where $m = \sizeof{\EQ}$.
Let $n$ denote the number of vertices in the original graph, i.e.,
the graph corresponding to $\LL$ in the earliest call to $\exactQuad$.
Let $\ncur$ denote the number of vertices in $\LL$ in the current call.
If $\ncur = 2$, the algorithm goes to Line~\ref{line:exacto1}, and hence we have $T(m) = O(1)$.
Otherwise, the algorithm goes to Lines~\ref{line:exactPartition}\,-\,\ref{line:exactRet}, among which
the most time-consuming work is eliminating $F$, inverting $\begin{pmatrix}
    			\matlow_{FF} & \matzero \\
    			\matlow_{CF} & \II_{CC}
	  		\end{pmatrix}$ and $\calD$, and recursively calling $\exactQuad$.
The first two both run in $O(\ncur^\omega)$ time, and the third runs in $2T(m/2)$ time.
When $m > n$, we can bound the number of vertices in the current call by $\ncur = O(n)$;
otherwise when $m \leq n$, we can bound the number of vertices in the current call by $\ncur = O(m)$.
Thus, We have
\begin{align}\label{eq:exactT}
	T(m) =
	\begin{cases}
	2T(m/2) + O(n^{\omega}) , & m > n, \\
	2T(m/2) + O(m^{\omega}) , & m \leq n.
	\end{cases}
\end{align}
Eq.~(\ref{eq:exactT}) leads to $T(m) = O(n^{\omega-1}m)$.

\end{proof}

\subsubsection{Approximating Quadratic Forms of $\vect{L}^\dag$ Upon Edge Deactivation} \label{sec:approx}




Clearly, if we only want to approximately compute the quadratic forms,
we can use the approximate partial Choleksy algorithm in Lemma~\ref{lem:partialChol} to speed up.
Thereby, we give an approximation algorithm 
$\quadEst(\LL, E, w, \zz, \theta,\epsilon)$, 
that computes an $\eps$-approximation of $\zz^\top \Lke^{\dag} \zz$
for a fixed $\zz \in \mathbb{R}^n$ for all $e\in \EQ$.
We also make a few modifications to maintain the error and further speed up.
We list the modifications in $\quadEst$ below:
\begin{enumerate}
	\item In step~\ref{item:Choleksy}, instead of computing the exact partial Cholesky factorization of $\LL$,
		we use the algorithm $\partialChol$ in Lemma~\ref{lem:partialChol} to obtain
		an approximate partial Cholesky factorization of $\LL$. However, if we pass the whole $\LL$ to $\partialChol$,
		it may change the edges in $E^{(1)}$, to which we need to perform $\theta$-deletions when deactivating them.
		Thus, instead, we first delete all edges in $E^(1)$ and pass the resulting $\LL$ to $\partialChol$,
		and then add those edges back to the approximate Schur complement $\SStil$ returned by it.
		This modification is feasible since
		adding edges with both endpoints in $C$ commutes with taking approximate partial Choleksy factorization
		(Lemma~\ref{lem:approxCholeksyAdd}).
		This modification is addressed on Lines~\ref{line:sparseStart}\,-\,\ref{line:sparseEnd} of Algorithm~\ref{alg:normEst}.
	\item By Lemma~\ref{lem:partialChol}, matrices $\calDtil$ and $\matlowtil$ returned by $\partialChol$ satisfy that
		$\calDtil$ is diagonal, $\matlowtil$ is sparse, and
		$\begin{pmatrix}
    		\matlowtil_{FF} & \matzero \\
    		\matlowtil_{CF} & \II_{CC}
  		\end{pmatrix}$
  		is a lower triangular matrix up to row exchanges.
  		Therefore, by applying inverses of diagonal matrix and lower triangular matrix quickly, we can compute
  		$
  			\yy = \begin{pmatrix} \yy_F \\ \yy_C \end{pmatrix} =
 			\begin{pmatrix}
    			\matlowtil_{FF} & \matzero \\
    			\matlowtil_{CF} & \II_{CC}
  			\end{pmatrix}^{-1}
 			\zz
 		$
 		and $\yy_F^\top \kh{\calDtil}^{-1} \yy_F$ in linear time of the number of nonzero entries.
 		This is addressed on
 		Lines~\ref{line:invertStart}\,-\,\ref{line:invertEnd} of Algorithm~\ref{alg:normEst}.
 	\item Since errors may accumulate among different levels of the recursion,
 		we bound the error by $\eps/\log \sizeof{\EQ}$ when taking approximate partial Cholesky factorization
 		(Line \ref{line:error1} of Algorithm~\ref{alg:normEst}), and
 		bound the error by $\eps - \eps/\log \sizeof{\EQ}$ when recursively calling $\quadEst$
 		(Line \ref{line:error2} of Algorithm~\ref{alg:normEst}).
 		Thereby, the errors add up to $\eps$ as required, and only an extra $\log^2 m$ factor is added to the running time
 		(see Lemma~\ref{lem:normEst} and its proof for details).
\end{enumerate}

According to the first modification, in this algorithm, we find $\yy$, $\DD$, and $\SS$ in step~\ref{item:Choleksy} by
taking approximate partial Cholesky factorization. The following Lemma shows how to find  them when an approximate
partial Cholesky factorization is given.

\begin{lemma}\label{lem:approxRecursion}
	For a graph $G = (V,E)$ with associate Laplacian $\LL$ and a set of vertices $C\subset V$.
	Let $F = V\setminus C$.
	Let $E^{(1)} \subset E$ be a set of edges with both endpoints in $C$,
	and $\LH$ be the Laplacian matrix corresponding to edges in $E^{(1)}$, i.e.,
	$\LH = \sum\nolimits_{e\in E^{(1)}} w(e)\bb_e \bb_e^\top$. 
	Clearly $\LL - \LH$ is also a Laplacian.
	Let an approximate partial factorization of $\LL - \LH$ be
	\begin{align}
		\LL - \LH \approx_\eps
		\begin{pmatrix}
    		\matlowtil_{FF} & \matzero \\
    		\matlowtil_{CF} & \II_{CC}
  		\end{pmatrix}
  		\begin{pmatrix}
			\calDtil & \matzero \\
			\matzero & \SStil
		\end{pmatrix}
		\begin{pmatrix}
		    \matlowtil_{FF} & \matzero \\
		    \matlowtil_{CF} & \II_{CC}
		\end{pmatrix}^\top.
	\end{align}
	Let $\yy = \begin{pmatrix} \yy_F \\ \yy_C \end{pmatrix} =
 			\begin{pmatrix}
    			\matlowtil_{FF} & \matzero \\
    			\matlowtil_{CF} & \II_{CC}
  			\end{pmatrix}^{-1}
 			\zz$ and $\SStil' = \SStil + \LH_{CC}$,
	then for each edge $e\in E^{(1)}$
	the following statement holds:
	\begin{align*}
		\zz^\top \Lke^\dag \zz
		\approx_\eps \yy_F^\top \kh{\calDtil^{-1}} \yy_F + \yy_C^\top \kh{ \SStil' \bsk e }^\dag \yy_C.
	\end{align*}
\end{lemma}

\begin{proof}
	By Lemma~\ref{lem:approxCholeksyAdd},
	adding edges with both endpoints in $C$ commutes with taking approximate partial Choleksy factorization.
	Thus, for each edge $e \in E^{(1)}$, by adding first edges in $E^{(1)} \setminus \setof{e}$
	and then the deactivated edge $e$ (i.e., edge $e$ with weight $\theta w(e)$),
	we have
	\begin{align*}
		\LL - \LH + \kh{ \LH \bsk e } \approx_{\eps}
		\begin{pmatrix}
    		\matlowtil_{FF} & \matzero \\
    		\matlowtil_{CF} & \II_{CC}
  		\end{pmatrix}
  		\begin{pmatrix}
			\calDtil & \matzero \\
			\matzero & \SStil + \kh{ \LH_{CC} \bsk e }
		\end{pmatrix}
 			\begin{pmatrix}
   			\matlowtil_{FF} & \matzero \\
   			\matlowtil_{CF} & \II_{CC}
		\end{pmatrix}^{\top}.
	\end{align*}
	Substituting $\kh{ \SStil' \bsk e } = \SStil + \kh{ \LH_{CC} \bsk e}$
	and $\LL\bsk e = \LL - \LH + \kh{ \LH \bsk e }$ leads to
	\begin{align}\label{eq:lkeapprox}
		\LL\bsk e \approx_{\eps}
		\begin{pmatrix}
    		\matlowtil_{FF} & \matzero \\
    		\matlowtil_{CF} & \II_{CC}
  		\end{pmatrix}
  		\begin{pmatrix}
			\calDtil & \matzero \\
			\matzero & \kh{ \SStil' \bsk e }
		\end{pmatrix}
 			\begin{pmatrix}
   			\matlowtil_{FF} & \matzero \\
   			\matlowtil_{CF} & \II_{CC}
		\end{pmatrix}^{\top}.
	\end{align}
	Note that $\SStil'$ is a Laplacian since it is a sum of two Laplacians.
	Inverting both sides of Eq.~(\ref{eq:lkeapprox}) leads to
	\begin{align}\label{eq:inverseApprox}
		\Lke^\dag &\approx_{\eps}
		\begin{pmatrix}
    		\matlowtil_{FF} & \matzero \\
    		\matlowtil_{CF} & \II_{CC}
  		\end{pmatrix}^{-\top}
  		\begin{pmatrix}
			\kh{\calDtil}^{-1} & \matzero \\
			\matzero & \kh{ \SStil' \bsk e }^\dag
		\end{pmatrix}
 			\begin{pmatrix}
   			\matlowtil_{FF} & \matzero \\
   			\matlowtil_{CF} & \II_{CC}
		\end{pmatrix}^{-1}.
	\end{align}
	Multiplying both sides of Eq.~(\ref{eq:inverseApprox}) by $\zz^\top$ on the left and $\zz$ on the right and
	substituting\\ $\yy = \begin{pmatrix} \yy_F \\ \yy_C \end{pmatrix} =\begin{pmatrix}
    			\matlowtil_{FF} & \matzero \\
    			\matlowtil_{CF} & \II_{CC}
  			\end{pmatrix}^{-1} \zz$, gives
	\begin{align}
  		\zz^\top \Lke^\dag \zz &\approx_{\eps}
  			\begin{pmatrix} \yy_F & \yy_C \end{pmatrix}
			\begin{pmatrix}
			\calDtil^{-1} & \matzero \\
			\matzero & \kh{ \SStil \bsk e }^\dag
			\end{pmatrix}
			\begin{pmatrix} \yy_F \\ \yy_C \end{pmatrix} \notag\\
		&=
			\yy_F^\top \kh{\calDtil^{-1}} \yy_F + \yy_C^\top \kh{ \SStil' \bsk e }^\dag \yy_C.
  	\end{align}
This completes the proof.
\end{proof}

The pseudocode for $\quadEst$ is given in Algorithm~\ref{alg:normEst}.
Its performance is characterized in Lemma~\ref{lem:normEst}. \\

\begin{restatable}[]{lemma}{thmnormEst}
	\label{lem:normEst}
	Given a connected undirected graph $G = (V,E)$ with $n$ vertices, $m$ edges, positive edge weights
	$w : E \to \rea_{+}$, and associated Laplacian $\LL$,
	a set of edges $\EQ\subset E$ such that every vertex in $V$ is incident with some edge $e \in \EQ$,
	a vector $\zz\in \mathbb{R}^n$,
  	and scalars $0 < \theta \leq 1/2$, $0<\eps\leq1/2$,
  	the algorithm $\quadEst(\LL, \EQ, w, \zz, \theta,\epsilon)$ returns a set of pairs
  	$\hat{N} = \{ (e,\hat{n}_e) \mid e \in \EQ \}$. With 	
  	high probability, the following statement holds:	For $\forall e \in \EQ$,
	\begin{align}\label{normErrorBound}
	n_e \approx_{\eps} \hat{n}_e,
	\end{align}
	where
	\[
		n_e = \zz^\top \Lke^\dag \zz.
	\]
	The total running time of this algorithm is bounded by
	$O(m \eps^{-2} \log^2 m \log^6 n \operatorname{polyloglog}(n))$.
\end{restatable}

\begin{algorithm}[H]
	\caption{$\quadEst(\LL, \EQ, w, \zz, \theta,\epsilon)$}
	\label{alg:normEst}
	\Input{
		$\LL$: A graph Laplacian.\\
		$\EQ$: A set of edges supported on vertices in $\LL$.\\
		$w$: An edge weight function.\\ 
		$\zz$: A vector whose dimension matches the number of vertices in $\LL$. \\
		$\theta$: An edge $e$'s weight should be temporarily reduced to $\theta w(e)$ when \\
			deactivating it.\\
		$\epsilon$: Error of the estimates. \\
	}
	\Output{
		$\hat{N} = \{(e,\hat{n}_e) \mid e \in \EQ\}$: $\hat{n}_e$ is an estimate of $n_e = \zz^\top \Lke^\dag \zz$.
	}
	
	Let $V$ denote the vertex set of $\LL$. \;
	\If{$|V| = 2$}{
		For every edge $e \in \EQ$, compute
		exact $\hat{n}_e = n_e = \zz^\top \left(\LL - (1 - \theta) w(e)\bb_e\bb_e^\top\right)^\dag \zz$,
		then combine the results and return $\hat{N} = \{(e,\hat{n}_e) \mid e \in \EQ\}$. \label{line:edgeCase}\;
	}
	Partition $\EQ$ into $E^{(1)}, E^{(2)}$ with $\sizeof{E^{(1)}} = \floor{\frac{\sizeof{\EQ}}{2}}$
		and $\sizeof{E^{(2)}} = \sizeof{\EQ} - \floor{\frac{\sizeof{\EQ}}{2}}$. \label{line:approxPart}\;
	\For{$i = 1$ to $2$}{
		Let $C$ denote endpoints of edges in $E^{(i)}$ and $F = V \setminus{C}$. \;
		Let $\LH$
			denote the Laplacian matrix corresponding to edges in $E^{(i)}$ \label{line:sparseStart}, i.e.,
			$\LH \gets \sum\nolimits_{e\in E^{(i)}} w(e)\bb_e\bb_e^\top$. \;
		$(\matlowtil, \calDtil, \SStil) \gets \partialChol(\LL - \LH, C, \epsilon/\log \sizeof{\EQ}) $ \label{line:error1}\;
		Add edges in $E^{(i)}$ back to $\SStil$ and store the resulting Laplacian in $\SStil'$,
		i.e., $\SStil' \gets \SStil + \LH_{CC}$ \label{line:sparseEnd}.  \;
		Compute $\yy = \begin{pmatrix} \yy_F \\ \yy_C \end{pmatrix} =\begin{pmatrix}
    			\matlowtil_{FF} & \matzero \\
    			\matlowtil_{CF} & \II_{CC}
  			\end{pmatrix}^{-1} \zz$ in linear time. \label{line:invertStart}\;
		Compute $f = \yy_F^\top \left(\calDtil \right)^{-1} \yy_F$ in linear time. \label{line:invertEnd} \;
		Call $\quadEst(\SStil', E^{(i)}, \yy_C, \theta, \epsilon - \epsilon/\log \sizeof{\EQ})$
		 to get an estimate $\hat{n}_e^{(i)}$ of $n_e^{(i)} = \yy_C^\top \kh{\SStil'\bsk e}^\dag \yy_C$ for all $e\in E^{(i)}$
		 and store $(e,f + \hat{n}_e^{(i)})$ in $\hat{N}^{(i)}$. \label{line:error2}\;
			
	}
	\KwRet{$\hat{N} = \hat{N}^{(1)} \cup \hat{N}^{(2)}$} \label{line:approxRet}
	
\end{algorithm}



\begin{proof}[Proof of Lemma~\ref{lem:normEst}]
\quad

	We first prove the error bound (i.e., Eq.~(\ref{normErrorBound})) by induction on the size of $\EQ$.
	
	For $\sizeof{\EQ} = 1$, we have $\sizeof{V} = 2$.
	Hence, the algorithm $\quadEst$ will go into Line~\ref{line:edgeCase} and returns an $\hat{n}_e = n_e$, which
	implies that $n_e \approx_{\eps} \hat{n}_e$ holds for any $\eps > 0$.
	
	Suppose Eq.~(\ref{normErrorBound}) holds for all $1\leq \sizeof{\EQ} \leq k$, $k\geq 1$.
	We now prove that it holds for $\sizeof{\EQ} = k + 1$, too.
	Clearly, by symmetry, it suffices to show that $n_e \approx_{\eps} \hat{n}_e$ holds for each $e\in E^{(1)}$.
	By Lemma~\ref{lem:partialChol}, matrices $\matlowtil, \calDtil, \SStil$ on Line~\ref{line:error1}
	satisfy
	\begin{align}\label{eq:LL-LH}
		\LL - \LH \approx_{\eps/\log \sizeof{\EQ}}
		\begin{pmatrix}
    		\matlowtil_{FF} & \matzero \\
    		\matlowtil_{CF} & \II_{CC}
  		\end{pmatrix}
  		\begin{pmatrix}
			\calDtil & \matzero \\
			\matzero & \SStil
		\end{pmatrix}
 			\begin{pmatrix}
   			\matlowtil_{FF} & \matzero \\
   			\matlowtil_{CF} & \II_{CC}
		\end{pmatrix}^{\top}.
	\end{align}
	By Lemma~\ref{lem:approxRecursion}, we have
  	\begin{align}\label{eq:approxF}
  		\zz^\top \Lke^\dag \zz &\approx_{\eps/\log \sizeof{\EQ}}
			\yy_F^\top \kh{\calDtil^{-1}} \yy_F + \yy_C^\top \kh{ \SStil' \bsk e }^\dag \yy_C,
  	\end{align}
  	where $\SStil' = \SStil + \LH_{CC}$ (Line~\ref{line:sparseEnd}).
  	Since $\sizeof{E^{(1)}} = \floor{\frac{\sizeof{\EQ}}{2}} \leq k$, by inductive assumption,
  	each $\hat{n}_e^{(1)}$ returned by the recursive call $\quadEst$ on Line~\ref{line:error2}
  	satisfies $\yy_C^\top \kh{ \hat{\SS} \bsk e }^\dag \yy_C \approx_{\eps - \eps/\log \sizeof{\EQ}} \hat{n}_e^{(1)} $,
  	which when adding $f = \yy_F^\top \kh{\calD^{-1}} \yy_F$ (Line~\ref{line:invertEnd}) to its both sides turns into
  	\begin{align}\label{eq:approxRecursion}
  		\yy_F^\top \kh{\calD^{-1}} \yy_F + \yy_C^\top \kh{ \hat{\SS} \bsk e }^\dag \yy_C
  		\approx_{\eps - \eps/\log \sizeof{\EQ}}
  		f + \hat{n}_e^{(1)}.
  	\end{align}
  	Combining Eq.~(\ref{eq:approxF}) and Eq.~(\ref{eq:approxRecursion}) and
  	substituting $n_e = \zz^\top \Lke^\dag \zz$,
  	we have $n_e \approx_{\eps} f + \hat{n}_e^{(1)}$. Thus, Eq.~(\ref{normErrorBound}) holds for $\sizeof{\EQ} = k + 1$, too.
  	By induction, it holds for all $\sizeof{\EQ}$.
  	
  	We then prove the running time of the algorithm. 
  	
  	Let $T(m, \eps)$ denote the running time of $\quadEst(\LL, \EQ, w, \zz, \theta,\epsilon)$,
  	where $m = \sizeof{\EQ}$ and $\eps$ is the error of estimates.
	Let $n$ denote the number of vertices in the original graph, i.e.,
	the graph corresponding to $\LL$ in the earliest call to $\quadEst$.
	Let $\ncur$ and $\mcur$ denote the number of vertices and the number of edges
	in $\LL$ in the current call, respectively.
	In each call other than the earliest call, the Laplacian $\LL$ equals $\SStil'$
	on Line~\ref{line:sparseEnd} of the parent call (i.e., the call that invoked current call), where
	we have the total number of edges in $\SS'$ being the number of edges in $\SStil$ plus the number of edges
	in $\LH$.
	Since $\LH$ is the Laplician corresponding to edges in $E^{(i)}$, which is precisely $\EQ$ in the current call,
	we have the number of edges in $\LH$ equaling $m = \sizeof{\EQ}$.
	By Lemma~\ref{lem:partialChol}, the number of edges in $\SStil$ is $O(\ncur \eps^{-2}\log^2 m \log n)$,
	where there is an extra $\log^2 m$ factor
	because the error is set to $\eps/\log \sizeof{\EQ}$ when calling $\partialChol$ on Line~\ref{line:error1}.
	Hence, the number of edges in $\LL$ in the current call is bounded by $\mcur = O(m + \ncur \eps^{-2}\log^2 m \log n)$.
	
	If $\ncur = 2$, the algorithm goes to Line~\ref{line:edgeCase}, and hence we have $T(m,\eps) = O(1)$.
	Otherwise, the algorithm goes to Lines~\ref{line:approxPart}\,-\,\ref{line:approxRet},
	among which the most time-consuming work can be divided into three parts:
	\begin{enumerate}
		\item The first part is computing $\begin{pmatrix}
    			\matlowtil_{FF} & \matzero \\
    			\matlowtil_{CF} & \II_{CC}
  			\end{pmatrix}^{-1}\zz$ and $\yy_F^\top \kh{\calDtil}^{-1}\yy_F$.
  			Since $\begin{pmatrix}
    			\matlowtil_{FF} & \matzero \\
    			\matlowtil_{CF} & \II_{CC}
  			\end{pmatrix}$ is a lower triangular matrix up to row exchanges and $\calDtil$ is diagonal,
  			their inverse can be applied in linear time of the number of nonzero entries.
  			By Lemma~\ref{lem:partialChol}, this part runs in
  			$O(\mcur + \ncur (\eps / \log m)^{-2}\log^3 n) = O(m + \ncur \eps^{-2}\log^2 m\log^3 n)$ time.
  		\item The second part is taking approximate partial Cholesky factorization, which by Lemma~\ref{lem:partialChol} runs in
  			$O(( \mcur \log^{3} n + \ncur (\eps / \log m)^{-2} \log^{5} n ) \operatorname{polyloglog}(n))
  			= \\
  			 O(m\log^3 n + \ncur \eps^{-2}\log^2 m \log^{5} n \operatorname{polyloglog}(n))$ time.
  		\item The third part is recursively calling
  			$\quadEst(\SStil', E^{(i)}, \yy_C, \theta, \epsilon - \epsilon/\log \sizeof{\EQ})$,
  			which runs in $2T(m/2, \eps - \eps / \log m)$ time.
	\end{enumerate}
	The first two parts add up to a running time of
	$O(m\log^3 n +  \ncur \eps^{-2}\log^2 m \log^{5} n \operatorname{polyloglog}(n))$.
	When $m > n$, we can bound the number of vertices in the current call by $\ncur = O(n)$;
	otherwise when $m \leq n$, we can bound the number of vertices in the current call by $\ncur = O(m)$.
	Thus, We have
	\begin{align}\label{eq:approxT}
		T(m,\eps) =
		\begin{cases}
			2T(m/2, \eps - \eps / \log m) + O(m\log^3 n +  n \eps^{-2}\log^2 m \log^{5} n \operatorname{polyloglog}(n)) , & m > n \\
			2T(m/2, \eps - \eps / \log m) + O(m\log^3 n +  m \eps^{-2}\log^2 m \log^{5} n \operatorname{polyloglog}(n)) , & m \leq n
		\end{cases}.
	\end{align}
	Eq.~(\ref{eq:approxT}) leads to $T(m,\eps) = O(m \eps^{-2}\log^2 m \log^{6} n \operatorname{polyloglog}(n))$.
\end{proof}

\subsection{Approximating $\mathcal{C}_\theta(e)$}\label{sec:centComp}

We are now ready to give the algorithm $\ECComp(G=(V,E), w, \theta,\epsilon)$,
which computes an $\eps$-approximation of the $\theta$-Kirchhoff edge centrality $\mathcal{C}_\theta(e)$
for all $e\in E$.  The pseudocode for $\ECComp$ is given in Algorithm~\ref{alg:ECComp}.
Its performance is characterized in Theorem~\ref{lem:edgeCentComp}.


\begin{algorithm}[H]
	\caption{$\ECComp(G=(V,E), w, \theta,\epsilon)$}
	\label{alg:ECComp}
	\Input{
		$G = (V,E)$, $w$: A connected undirected graph with positive edges \\weights $w: E\to \rea_{+}$. \\
		$\theta$: An edge $e$'s weight should be temporarily reduced to $\theta w_e$ when\\ deactivating it. \\
		$\epsilon$: Error of the centrality estimate per edge.
	}
	\Output{
		$\hat{C} = \left\{ (e, \hat{c}_e) \mid e \in E \right\}$: $\hat{c}_e$ is an estimate of $\mathcal{C}_\theta(e)$,
			the $\theta$-Kirchhoff\\ edge centrality of $e$.
	}
	Let $\zz_1,\ldots,\zz_M$ be independent random $\pm 1$ vectors, where $M = \ceil{192\eps^{-2}\ln(2n)}$. \;
	\For{$i = 1$ to $M$}{
		Call $\quadEst(\LL^G,E, w, \zz_i, \theta, \frac{\epsilon}{2})$
			to get an estimate $\hat{n}_e$ of $n_e = \zz_i^\top \Lke^\dag \zz_i$ for each $e \in E$,
			and	store each $\hat{n}_e$ in $\hat{n}_e^{(i)}$ \label{line:nej}.
	}
	For each $e\in E$ compute $\hat{c}_e = \frac{n}{M} \sum\limits_{i=1}^M \hat{n}_e^{(i)}$ and
		return $\hat{C} = \left\{ (e, \hat{c}_e) \mid e \in E \right\}$. \label{line:estce}
\end{algorithm}

\begin{proof}[Proof of Theorem~\ref{lem:edgeCentComp}]
	\quad
	
	The running time is the total cost of $O(\eps^{-2} \log n)$ calls to $\quadEst$, each of which runs in
	$O(m \eps^{-2} \log^2 m \log^6 n \operatorname{polyloglog}(n))$ time according to Lemma~\ref{lem:normEst}.
	
	Since $M = \ceil{192\eps^{-2}\ln(2n)} \geq 48 \kh{\frac{\eps}{2}}^{-2}\ln(2n)$,
	by Lemma~\ref{lem:mcL}, for each $e\in E$, there is
	\begin{align*}
		\frac{1}{M} \sum\limits_{i=1}^M \zz_i^\top \Lke^{\dag} \zz_i \approx_{\frac{\eps}{2}} \trace{\LL\bsk e}.
	\end{align*}
	 Multiplying both sides by $n$
	 and substituting $\mathcal{C}_\theta(e) = n \trace{\LL\bsk e}$, we have
	\begin{align}\label{eq:ceapproxnorms}
		\frac{n}{M} \sum\limits_{i=1}^k \zz_i^\top \Lke^{\dag} \zz_i \approx_{\frac{\eps}{2}} \mathcal{C}_\theta(e).
	\end{align}
	
	By Lemma~\ref{lem:normEst}, each $\hat{n}_e^{(i)}$ on Line~\ref{line:nej} satisfies
	\begin{align}\label{eq:nej}
		n_e^{(i)} \approx_{\frac{\eps}{2}} \hat{n}_e^	{(i)},
	\end{align}
	where $n_e^{(i)} = \zz_i^\top \Lke^\dag \zz_i$.
	Summing Eq.~(\ref{eq:nej}) over $i = 1,\ldots,M$ and multiplying both sides by $\frac{n}{M}$ lead to
	\begin{align*}
		\frac{n}{M} \sum\limits_{i=1}^M \zz_i^\top \Lke^{\dag} \zz_i \approx_{\frac{\eps}{2}}
		\frac{n}{M} \sum\limits_{i=1}^M \hat{n}_e^{(i)}.
	\end{align*}
	Combining with Eq.~(\ref{eq:ceapproxnorms}) and substituting $\hat{c}_e = \frac{n}{M} \sum\limits_{i=1}^M \hat{n}_e^{(i)}$
	(Line~\ref{line:estce}) lead to $\mathcal{C}_\theta(e) \approx_\eps \hat{c}_e$.
\end{proof}

\input{SM.tex}
\input{VertexWoodbury.tex}

\input{conclusion}


\newpage

\input{centrality.bbl}


\newpage

\begin{appendix}
\input{Errors.tex}

\end{appendix}

\end{document}

%% file: SM.tex
\section{Algorithm for Approximating $\theta$-Kirchhoff Edge Centrality $\mathcal{C}_\theta^\Delta(e)$}
\label{sec:edgebyjl}

By Sherman-Morrison formula, for an edge $e\in E$ and a scalar $0 < \theta < 1$,
	we have
	\begin{align}\label{eq:morrison}
		\kh{\LL\bsk e}^\dag = \kh{\LL - (1 - \theta)w(e)\bb_e\bb_e^\top}^\dag
		= \LL^\dag + (1 - \theta) \frac{w(e) \LL^\dag \bb_e \bb_e^\top \LL^\dag}{1 - (1 - \theta)w(e)\bb_e^\top \LL^\dag \bb_e}.
	\end{align}
Since the off-the-shelf Sherman-Morrison formula is
for full rank matrices,
we give the detailed proof of Equation~(\ref{eq:morrison}) in
Appendix~\ref{sec:morrison}.

Since by our definition $\mathcal{C}_\theta^\Delta(e) = \Kf{G\bsk e} - \Kf{G} = n\kh{ \trace{\kh{\LL\bsk e}^\dag} - \trace{\LL^\dag}}$,
it follows that
\begin{align}\label{eq:SM}
	\mathcal{C}_\theta^\Delta(e) = 
	n (1 - \theta) \frac{w(e) \trace{\LL^\dag \bb_e \bb_e^\top \LL^\dag}}{1 - (1 - \theta)w(e) \bb_e^\top \LL^\dag \bb_e}.
\end{align}

The numerator of~(\ref{eq:SM}) is the trace of an implicit matrix, and hence can be approximated by
Hutchinson's~\cite{AT11,Hut89} Monte-Carlo method.
To apply $\LL^\dag$, we can utilize nearly-linear time solvers for Laplacian systems~\cite{ST14,CKMPPRX14}.
We will use the solver from~\cite{CKMPPRX14}, whose
performance can be characterized in the following lemma.

\begin{lemma}[Theorem 1.1 of~\cite{CKMPPRX14}, paraphrased]
	\label{lem:laplsolver}
	There is an algorithm \\$\yy = \LaplSolver(\LL^G,\zz,\delta)$
	which takes a Laplacian matrix $\LL^G$ of
	a graph $G$ with $n$ vertices and $m$ edges,
	a vector $\zz \in \rea^n$,
	and a scalar $\delta > 0$,
	and returns a vector $\yy \in \rea^n$ such that
	with high probability the following statement holds:
	\begin{align*}
		\len{\yy - \LL^\dag \zz}_{\LL} \leq \delta
		\len{\LL^\dag \zz}_{\LL},
	\end{align*}
	where $\len{\xx}_{\LL} = \sqrt{\xx^\top \LL \xx}$.
	The algorithm runs in expected time
	$O(m\log^{0.5}n\log (1/\delta)\operatorname{polyloglog}(n))$.
\end{lemma}

To track the error for the solver, we will need
the following two lemmas,
whose proofs are deferred to Appendix~\ref{sec:Errorsedge}.

\begin{lemma}\label{lem:edgeSolveError}
	Let $\LL$ be the Laplacian of a graph with all weights
	in the range $[1,U]$, and $\zz$ be any vector
	such that $\len{\zz}^2 \leq n$.
	Suppose $\yy$ is a vector such that
	$\len{\yy - \LL^\dag \zz}_{\LL} \leq
	\delta \len{\LL^\dag \zz}_{\LL}$ for some
	$0 < \delta < 1$.
	For any edge $e$ of the graph,
	we have
	\begin{align}\label{eq:edgeSolveError}
		\left|
		\yy^\top \bb_e \bb_e^\top \yy - 
		\zz^\top \LL^\dag \bb_e \bb_e^\top \LL^\dag \zz
		\right|
		\leq 6\delta n^5 U^2.
	\end{align}
\end{lemma}

\begin{lemma}\label{lem:edgeLower}
	Let $\LL$ be the Laplacian of a graph with all weights
	in the range $[1,U]$.
	For any edge $e$ of the graph, we have
	\begin{align*}
		\trace{\LL^\dag \bb_e \bb_e^\top \LL^\dag}
		\geq \frac{2}{n^2 U^2}.
	\end{align*}
\end{lemma}

The denominator of~(\ref{eq:SM}) is just $1 - (1 - \theta) w(e)\mathcal{R}_{\mathrm{eff}}(e)$.
Since $w(e)\mathcal{R}_{\mathrm{eff}}(e)$ is between 0 and 1 and $\theta$ is positive,
$(1 - \theta) w(e)\mathcal{R}_{\mathrm{eff}}(e)$ is strictly bounded away from $1$.
Thus, we can multiplicatively approximate the denominator by approximating $\mathcal{R}_{\mathrm{eff}}(e)$,
for which we can use the random projection in~\cite{SS11}.
By Using the solvers from~\cite{CKMPPRX14}
in the effective resistance estimation
procedure of~\cite{SS11}, we immediately have the following lemma:

\begin{lemma}\label{lem:ERest}
	There is an algorithm $\textsc{EREst}(G, \eps)$ that when given a graph $G = (V,E)$,
	returns an estimate $\hat{r}_e$ of $\mathcal{R}_{\mathrm{eff}}(e)$ for all $e\in E$ in
	$O(m \eps^{-2} \log^{2.5}n
	\operatorname{polyloglog}(n))$ time.
	With high probability, $\hat{r}_e \approx_{\eps} \mathcal{R}_{\mathrm{eff}}(e)$ holds
	for all $e\in E$.
\end{lemma}
As these estimates are approximate, we will need
to bound their approximations when subtracted from $1$.
Here we use the fact that $0 < \theta < 1$, and that
the weight times effective resistance of an edge,
$w(e) \bb_e^T \LL^{\dag} \bb_e$ is between $0$ and $1$.
Since we will also need the matrix version of this type of
approximation propagation when computing Kirchhoff vertex
centralities in Section~\ref{sec:vertexbyjlschur},
we will state the more general version here.
\begin{lemma}
\label{lem:SubtractError}
If $\AA$ and $\BB$ are matrices such that
$0 \preceq \AA \preceq \II$,
and $\AA \approx_{\epsilon} \BB$ for some $0 < \epsilon \leq 1/2$,
then for any $0 < \theta \leq 1/2$ such that $\epsilon / \theta \leq 1/10$,
we have
\[
\II - \left( 1 - \theta\right) \AA
\approx_{3\epsilon / \theta}
\II - \left( 1 - \theta\right) \BB.
\]
\end{lemma}
The proof is deferred to Appendix~\ref{sec:Errors}.

We then give an algorithm $\ECComptwo$ to approximate the $\theta$-edge
Kirchhoff centrality $\mathcal{C}_\theta^\Delta$
for all $m\in E$.
The pseudocode for $\ECComptwo$ is given in Algorithm~\ref{alg:eccomptwo}.
The performance
of $\ECComptwo$ is characterized in Theorem~\ref{lem:edgebyjl}.

\begin{algorithm}[H]
	\caption{$\ECComptwo(G, \theta,\epsilon)$}
	\label{alg:eccomptwo}
	\Input{
		$G$: A graph. \\
		$\theta$: a scalar between 0 and $1/2$. \\
		$\epsilon$: the error parameter. \\
	}
	\Output{
		$\hat{C}^\Delta
		= \{ (e,\hat{c}_e^\Delta) \mid e \in E \}$.
	}
	Let $\zz_1,\ldots,\zz_M$ be independent random $\pm 1$ vectors, where $M = \ceil{432\eps^{-2}\ln(2n)}$. \;
	\For{$i = 1$ to $M$}{
		$\yy_i \gets \LaplSolver(\LL^G, \zz_i, \frac{1}{36}\eps n^{-7} U^{-4})$ \;
		\For{each $e \in E$}{
			Compute $\hat{c}_e^{\Delta(i)} = \yy_i^\top \bb_e \bb_e^\top \yy_i$. \;
		}
	}
	$\hat{r}_e \gets \textsc{EREst}(G, \theta\eps/9)$ \;
	Compute $\hat{c}_e^\Delta = (1 - \theta) \frac{n}{M} w(e) \sum\limits_{i}^{M} \hat{c}_e^{\Delta(i)} / (1 - (1 - \theta)w(e)\hat{r}_e$) for each $e$.
\end{algorithm}

\begin{proof}[Proof of Theorem~\ref{lem:edgebyjl}]
	The running time is the total cost of
	$O(\eps^{-2} \log n)$ calls to $\LaplSolver$ each of
	which runs in
	$O(m\log^{1.5}n\log (1/\eps)\operatorname{polyloglog}(n))$ time, and a call to $\textsc{EREst}$ which runs in
	$O(m \theta^{-2} \eps^{-2} \log^{2.5}n
	\operatorname{polyloglog}(n))$ time.
	
	Since $M = \ceil{432\eps^{-2}\ln(2n)} \geq 48 \kh{\eps/3}^{-2} \ln(2n)$, by Lemma~\ref{lem:mcL}, we have
	\begin{align}\label{eq:er1}
		\frac{1}{M} \sum\limits_{i=1}^M \zz_i^\top \LL^\dag \bb_e \bb_e^\top \LL^\dag \zz_i \approx_{\eps/3} \trace{\LL^\dag \bb_e \bb_e^\top \LL^\dag}.
	\end{align}
	By Lemma~\ref{lem:edgeLower}, we have
	\begin{align*}
		\trace{\LL^\dag \bb_e \bb_e^\top \LL^\dag}
		\geq \frac{2}{n^2 U^2},
	\end{align*}
	and hence
	\begin{align}\label{eq:trlower}
		\frac{1}{M} \sum\limits_{i=1}^M \zz_i^\top \LL^\dag \bb_e \bb_e^\top \LL^\dag \zz_i
		\geq \exp(-\eps/3) \frac{2}{n^2 U^2}
		\geq \frac{1}{n^2 U^2},
	\end{align}
	where the second inequality follows by $0 < \eps \leq 1/2$.
	
	Since we set $\delta = \frac{1}{36}\eps n^{-7} U^{-4}$ when invoking
	$\LaplSolver$, by Lemma~\ref{lem:laplsolver},
	\begin{align*}
		\len{\yy_i - \LL^\dag \zz_i}_{\LL} \leq
		\frac{1}{36}\eps n^{-7} U^{-4}
		\len{\LL^\dag \zz_i}_{\LL}
	\end{align*}
	holds for each $i$.
	Then, by Lemma~\ref{lem:edgeSolveError}, we have that
	\begin{align*}
		\left|
		\yy_i^\top \bb_e \bb_e^\top \yy_i - 
		\zz_i^\top \LL^\dag \bb_e \bb_e^\top \LL^\dag \zz_i
		\right|
		\leq \frac{1}{6}\eps n^{-2} U^{-2}
	\end{align*}
	holds for each $i$.
	We then have
	\begin{align*}
		&\left|
		\frac{1}{M} \sum\limits_{i=1}^M \yy_i^\top \bb_e \bb_e^\top \yy_i -
		\frac{1}{M} \sum\limits_{i=1}^M \zz_i^\top \LL^\dag \bb_e \bb_e^\top \LL^\dag \zz_i \right| \\
		\leq &
		\frac{1}{M} \sum\limits_{i=1}^M
		\left| \yy_i^\top \bb_e \bb_e^\top \yy_i
		- \zz_i^\top \LL^\dag \bb_e \bb_e^\top \LL^\dag \zz_i \right| \\
		\leq & \frac{1}{6}\eps n^{-2} U^{-2} \\
		\leq  &\frac{1}{6}\eps \kh{
		\frac{1}{M} \sum\limits_{i=1}^M \zz_i^\top \LL^\dag \bb_e \bb_e^\top \LL^\dag \zz_i },
	\end{align*}
	where the last inequality follows by~(\ref{eq:trlower}).
	Thus,
	\begin{align*}
		(1 - \eps/6)
		\frac{1}{M} \sum\limits_{i=1}^M \zz_i^\top \LL^\dag \bb_e \bb_e^\top \LL^\dag \zz_i
		\leq
		\frac{1}{M}\sum\limits_{i=1}^M \yy_i^\top \bb_e \bb_e^\top \yy_i
		\leq
		(1 + \eps/6)
		\frac{1}{M} \sum\limits_{i=1}^M \zz_i^\top \LL^\dag \bb_e \bb_e^\top \LL^\dag \zz_i,
	\end{align*}
	which implies
	\begin{align}\label{eq:er2}
		\frac{1}{M} \sum\limits_{i=1}^M \zz_i^\top \LL^\dag \bb_e \bb_e^\top \LL^\dag \zz_i \approx_{\eps/3}
		\frac{1}{M}\sum\limits_{i=1}^M \yy_i^\top \bb_e \bb_e^\top \yy_i.
	\end{align}
	
	By Lemma~\ref{lem:ERest}, we have $\hat{r}_e \approx_{\theta\eps/9} \mathcal{R}_\mathrm{eff}(e)$,
	which by Lemma~\ref{lem:SubtractError} implies
	\begin{align}\label{eq:er3}
		1 - (1 - \theta)w(e)\hat{r}_e \approx_{\eps/3} 1 - (1 - \theta) w(e)\mathcal{R}_\mathrm{eff}(e).
	\end{align}
	Combining Equations~(\ref{eq:er1}),~(\ref{eq:er2}) and~(\ref{eq:er3}), we have
	\begin{align*}
		\frac{\frac{1}{M} \sum\limits_{i=1}^M \yy_i^\top \bb_e \bb_e^\top \yy_i}{1 - (1 - \theta)w(e)\hat{r}_e}
		\approx_\eps 
		\frac{\trace{\LL^\dag \bb_e \bb_e^\top \LL^\dag}}{1 - (1 - \theta)w(e)\mathcal{R}_\mathrm{eff}(e)},
	\end{align*}
	which together with Equation~(\ref{eq:SM}) proves this theorem.
\end{proof}

%% file: VertexWoodbury.tex
\section{Algorithm for Approximating $\theta$-Kirchhoff Vertex Centrality} 
\label{sec:vertexbyjlschur}

We now combine the projection based approximation
algorithm from Section~\ref{sec:edgebyjl} with the recursive
Schur complement approximation algorithm to produce
a routine for estimating Kirchhoff vertex centrality
as defined in Definition~\ref{def:VertexRemoval} in nearly-linear time.

\subsection{Turning to Low-rank Updates}

We will treat the $\theta$-deletion of a vertex as $\theta$-deleting a batch of
edges from the graph, which in turn corresponds to a high rank
update to the graph Laplacian.
Specifically, we can define the matrix
$\BB_{\Ev{v}}$
as the $\degu{v} \times n$ edge-vertex incidence matrix containing the edges
incident to $v$,
and $\WW_{\Ev{v}}$
as the corresponding $\degu{v} \times \degu{v}$ diagonal edge weight
matrix.
Here we use $\Ev{v} \defeq \setof{(u,v) \,|\, u\sim v}$ to denote
the set of edges incident with $v$,
and $\degu{v} \defeq \sizeof{\Ev{v}}$ to denote
the number of edges incident with $v$.
Then the graph Laplacian with edges in $\Ev{v}$ $\theta$-deleted is
\[
\LL \bsk \Ev{v}
= \LL - \left( 1 - \theta\right) \BB_{\Ev{v}}^\top \WW_{\Ev{v}} \BB_{\Ev{v}}.
\]
By Lemma~\ref{lem:mcL}, our goal 
becomes solving
Problem~\ref{prob:norm} on the difference between the pseudoinverses
of these matrices.
Specifically, computing the value of
\[
\zz^\top
\left(\left( \LL - \left( 1 - \theta\right) \BB_{\Ev{v}}^\top\WW_{\Ev{v}} \BB_{\Ev{v}} \right)^{\dag}
- \LL^{\dag} \right)
\zz
\]
for a vector $\zz$.
For this we once again turn to low-rank updates,
specifically the Woodbury formula.
\begin{lemma}[Derived from Woodbury formula]
\label{lem:Woodbury}
	Given an edge set $T \subset E$ supported on vertex set $C$,
	and a scalar $0 < \theta < 1$.
	Let $\BB_T$ be the ${\sizeof{T} \times n}$ edge-vertex incidence matrix corresponding to edges in $T$,
	and $\WW_T$ be the $\sizeof{T}\times \sizeof{T}$ diagonal edge weight matrix corresponding to edges in $T$.
The following statement holds:
\begin{align}\label{eq:woodbury}
\kh{\LL \bsk T}^\dag =
\LL^\dag + \left( 1 - \theta \right) \LL^\dag \BB_{T}^\top \WW_{T}^{1 / 2}
\kh{ \II - (1 - \theta)\WW_{T}^{1 / 2} \BB_{T}
  \LL^\dag
\BB_{T}^\top\WW_{T}^{1 / 2} }^{-1}
\WW_{T}^{1 / 2} \BB_{T} \LL^\dag.
\end{align}
\end{lemma}
Since the off-the-shelf Woodbury formula is
for full rank matrices,
we give the detailed proof of Equation~(\ref{eq:woodbury}) in
Appendix~\ref{sec:morrison}.

Note that this formula applies to any subset of $T$ edges.
The only property of evaluating Kirchhoff vertex centrality we need is that
the total size of such $T$s over all vertices is $2m$.

This means just as in Section~\ref{sec:edgebyjl},
the problem reduces to estimating
\begin{align*}
\zz^\top
\LL^\dag \BB_{T}^\top \WW_{T}^{1 / 2}
\kh{ \II - (1 - \theta)\WW_{T}^{1 / 2} \BB_{T}
  \LL^\dag
\BB_{T}^\top\WW_{T}^{1 / 2} }^{-1}
\WW_{T}^{1 / 2} \BB_{T} \LL^\dag \zz.
\end{align*}
Furthermore, since we can compute $\LL^{\dag} \zz$
to high accuracy via a single solver to a linear system
in a graph Laplacian, and $\WW_{T}^{1/2} \BB_{T}$
is $\abs{T} \times n$ matrix with $2\abs{T}$ nonzero entries,
we can compute for each set $T$ the vector
$\WW_{T}^{1/2} \BB_{T} \LL^{\dag} \zz$
in $O(\abs{T})$ time (after $\tilde{O}(m)$ preprocessing time
to compute an approximation to $\LL^{\dag} \zz$).
To track the error for the solver, we will need to following two lemmas, which we prove in Appendix~\ref{sec:Errorsvertex}.
\begin{lemma}
\label{lem:SolverError}
Let $\LL$ be the Laplacian of a graph with all weights
	in the range $[1,U]$, and $\zz$ be any vector
	such that $\len{\zz}^2 \leq n$.
	Suppose $\yy$ is a vector such that
	$\len{\yy - \LL^\dag \zz}_{\LL} \leq
	\delta \len{\LL^\dag \zz}_{\LL}$ for some
	$0 < \delta < 1$.
For any edge set $T\subset E$, we have:
\begin{multline}\label{eq:SolverError}
\left|\zz^\top
\LL^\dag \BB_{T}^\top \WW_{T}^{1 / 2}
\kh{ \II - (1 - \theta)\WW_{T}^{1 / 2} \BB_{T}
  \LL^\dag
\BB_{T}^\top\WW_{T}^{1 / 2} }^{-1}
\WW_{T}^{1 / 2} \BB_{T} \LL^\dag \zz \right.\\
-
\left.\yy ^\top\BB_{T}^\top \WW_{T}^{1 / 2}
\kh{ \II - (1 - \theta)\WW_{T}^{1 / 2} \BB_{T}
  \LL^\dag
\BB_{T}^\top\WW_{T}^{1 / 2} }^{-1}
\WW_{T}^{1 / 2} \BB_{T} \yy \right|
\leq 6\theta^{-1}\delta n^5 U^2.
\end{multline}
\end{lemma}

\begin{lemma}\label{lem:Lower}
	Let $\LL$ be the Laplacian of a graph with all weights
	in the range $[1,U]$.
	For any edge set $T\subset E$ of the graph, we have
	\begin{align*}
		\trace{\LL^\dag \BB_{T}^\top \WW_{T}^{1 / 2}
\kh{ \II - (1 - \theta)\WW_{T}^{1 / 2} \BB_{T}
  \LL^\dag
\BB_{T}^\top\WW_{T}^{1 / 2} }^{-1}
\WW_{T}^{1 / 2} \BB_{T} \LL^\dag}
		\geq \frac{2\sizeof{T}}{n^2U^2}.
	\end{align*}
\end{lemma}

\subsection{Approximating Quadratic Forms}

Since we can utilize nearly-linear time solvers for
Laplacian linear systems
to compute high accuracy approximations to the vector $\LL^\dag \zz$,
the problem is further reduced to estimating quadratic forms of
\begin{align*}
\kh{ \II - (1 - \theta)\WW_{T}^{1 / 2} \BB_{T}
  \LL^\dag
\BB_{T}^\top \WW_{T}^{1 / 2} }^{-1}.
\end{align*}

Since the edges in $T$ form a subgraph of $\LL$,
we have $\BB_T^\top \WW_T \BB \preceq \LL$,
and in turn
\begin{align*}
0 \preceq
\WW_{T}^{1 / 2} \BB_{T}
  \LL^\dag
\BB_{T}^\top\WW_{T}^{1 / 2}
\preceq \II
\end{align*}
This coupled with the assumption that $0 < \theta$
means that the eigenvalues of matrix\\
$(1 - \theta)\WW_{T}^{1 / 2} \BB_{T}
  \LL^\dag
\BB_{T}^\top \WW_{T}^{1 / 2}$
are bounded away from 1.
Therefore, we can use iterative methods to solve the
resulting system.
As we work entirely with matrix approximations,
we will use the following matrix-based version
of Chebyshev iteration.
More details on these iterative methods can be found
in Section 11.2 of~\cite{GVL12}.
\begin{lemma}[Chebyshev iteration]
\label{lem:Chebyshev}
There is an algorithm $\ChebSolve(\PP, \kappa, \epsilon, \bb)$
such that for any positive definite matrix $\PP$
along with $\kappa$ such that
$\frac{1}{\kappa} \II \preceq \PP \preceq \II$,
$\ChebSolve(\PP, \kappa, \epsilon, \bb)$ corresponds
to a linear operator on $\bb$ such that the matrix $\ZZ_{\ChebSolve}$ realizing
this operator satisfies
\begin{align*}
\ZZ_{\ChebSolve} \approx_{\epsilon} \PP^{-1},
\end{align*}
and the cost of the algorithm is $O(\sqrt{\kappa} \log(1/\eps))$
matrix-vector multiplications involving $\PP$.
\end{lemma}

Therefore the main difficulty becomes finding the matrix
\begin{align*}
\WW_{T}^{1 / 2} \BB_{T}
  \LL^\dag
\BB_{T}^\top\WW_{T}^{1 / 2}.
\end{align*}
Note that while $\BB_{T}$ has up to $n$ columns, most of these
column are $0$s.
So it means that we can only consider the entries
corresponding to $V(T)$, the set of vertices
incident to at least one edge in $T$, 
using the following fact about Schur complements.
\begin{fact}
\label{fact:SchurSubset}
Let $\LL$ be a Laplacian matrix, and $C$ be a subset of vertices. Then,
we have
\begin{align*}
\left( \LL^{\dag} \right)_{CC}
=  \SC\left(\LL, C\right)^{\dag}.
\end{align*}
\end{fact}
However, we only have approximate Schur complements.
To bound this also, we once again invoke the bound about
preservations of approximations when subtracting matrices
from $\II$ from Lemma~\ref{lem:SubtractError}.

\begin{lemma}
\label{lem:UsingApproxSchur}
	There is an algorithm $\widetilde{x} = \QuadSolver(\BB_T,\WW_T, \bb, \theta, \eps, \SStil)$ which takes
	an edge-vertex incidence matrix $\BB_T$ corresponding to edges in $T\subset E$
	with edge weight matrix $\WW_T$ supported on vertex set $V(T)$,
	a vector $\bb \in \mathbb{R}^n$,
	scalars $0 < \theta \leq 1/2$ and $0 < \eps < 1/2$,
	and a Laplacian matrix $\SStil$ whose edges are supported on $V(T)$
	such that $\SStil \approx_{\eps \theta / 9} \SC(\LL, V(T))$,
	and returns a value $\widetilde{x}$ satisfying
\begin{align*}
\widetilde{x}
\approx_{\epsilon}
\bb^\top
\kh{\II - (1 - \theta)\WW^{1/2}_T\BB_T\LL^\dag\BB_T^\top\WW^{1/2}_T}^{-1}
\bb.
\end{align*}
The algorithm runs in time
$O(\mathrm{nnz}(\SStil)\theta^{-0.5}\log^3 n \log(1/\eps) + \sizeof{T}\theta^{-2.5}\eps^{-2}\log^5 n \log(1/\eps) \operatorname{polyloglog}(n))$,
where $\mathrm{nnz}(\SStil)$ is the number of nonzero entries in $\SStil$.
\end{lemma}

\begin{proof}
We will invoke preconditioned Chebyshev iteration as
stated in Lemma~\ref{lem:Chebyshev} to estimate the
quantity
\begin{align*}
\bb^\top
\kh{\II - (1 - \theta)
\WW^{1/2}_{T}\BB_{T,V(T)}
  \SStil^\dag
\BB_{T,V(T)}^\top\WW_{T}^{1/2}}^{-1}
\bb.
\end{align*}
Since $\BB_{T}$ is only non-zero on the entries
corresponding to $V(T)$,
Fact~\ref{fact:SchurSubset} gives
\begin{align*}
\WW^{1/2}_T\BB_T\LL^\dag\BB_T^\top\WW^{1/2}_T
& = \WW^{1/2}_T\BB_{T,V\left(T\right)}
  \SC\left( \LL, V\left(T\right) \right)^\dag
\BB_{T,V\left(T\right)}^\top\WW^{1/2}_T,
\end{align*}
and hence Fact~\ref{fact:Approximations} Part~\ref{part:CompositionMatrix}
gives
\begin{align}
\WW^{1/2}_T\BB_T\LL^\dag\BB_T^\top\WW^{1/2}_T
& \approx_{\epsilon \theta/9}
\WW^{1/2}_T\BB_{T, V\left( T \right)}\SStil^\dag\BB_{T, V\left( T \right)}^\top\WW^{1/2}_T.
\label{eq:ReplaceWithApproxSchur}
\end{align}
Also, since $T$ is a subset of edges,
\begin{align*}
\left( \WW_T^{1/2} \BB_T \right)^{\top}
\left( \WW_T^{1/2} \BB_T  \right)
\preceq \LL,
\end{align*}
which in turn implies
\begin{align*}
\WW_T^{1/2} \BB_T 
\LL^{\dag}
\BB_{T}^\top \WW_T^{1/2}
\preceq \II,
\end{align*}
and
\begin{align*}
\theta \II
\preceq \II - (1 - \theta)\WW^{1/2}_T\BB_{T,V(T)}
\SC\left( \LL, V\left(T\right) \right)^\dag
\BB_{T,V(T)}^\top\WW^{1/2}_T
\preceq \II.
\end{align*}

Combining this with the approximation factor above
from Equation~(\ref{eq:ReplaceWithApproxSchur}) and
Lemma~\ref{lem:SubtractError} then gives
\begin{align}\label{eq:quadapx1}
\II - 
(1 - \theta) \WW^{1/2}_T\BB_{T}\LL^\dag\BB_{T}^\top\WW^{1/2}_T
\approx_{\epsilon / 3}
\II - 
(1 - \theta) \WW^{1/2}_T\BB_{T,V(T)}\SStil^\dag\BB_{T,V(T)}^\top\WW^{1/2}_T.
\end{align}
To apply $\SStil^\dag$, we can invoke the algorithm
$\partialChol(\SStil,\setof{v},\eps\theta/9)$
in Lemma~\ref{lem:partialChol}
for an
arbitrary vertex $v$ to get an
$\eps\theta/9$-approximate
sparse \textit{complete} Cholesky factorization of $\SStil$ and then
apply its inverse quickly.
Suppose the Cholesky factorization returned is
$\matlowtil \calDtil \matlowtil^\top \approx_{\eps\theta/9} \SStil$,
then again by Lemma~\ref{lem:SubtractError}
we have
\begin{align}\label{eq:quadapx2}
\II - 
(1 - \theta) \WW^{1/2}_T\BB_{T,V(T)}\SStil^\dag\BB_{T,V(T)}^\top\WW^{1/2}_T
\approx_{\epsilon / 3}
\II - 
(1 - \theta) \WW^{1/2}_T\BB_{T,V(T)}\kh{\matlowtil \calDtil \matlowtil^\top}^{\dag}\BB_{T,V(T)}^\top\WW^{1/2}_T.
\end{align}
Combining Equation~(\ref{eq:quadapx1}) and~(\ref{eq:quadapx2}) leads to
\begin{align}\label{eq:quadapx3}
	\II - 
(1 - \theta) \WW^{1/2}_T\BB_{T}\LL^\dag\BB_{T}^\top\WW^{1/2}_T
	\approx_{2\eps/3}
	\II - 
(1 - \theta) \WW^{1/2}_T\BB_{T,V(T)}\kh{\matlowtil \calDtil \matlowtil^\top}^{\dag}\BB_{T,V(T)}^\top\WW^{1/2}_T,
\end{align}
which also means that all the eigenvalues of
$
\II - 
(1 - \theta) \WW^{1/2}_T\BB_{T,V(T)}\kh{\matlowtil \calDtil \matlowtil^\top}^{\dag}\BB_{T,V(T)}^\top\WW^{1/2}_T$
are between $\mathrm{exp}(-2\eps/3)\theta$ and 1.
Therefore by Lemma~\ref{lem:Chebyshev}, we can access
a linear operator $\ZZ_{\textsc{Solve}}$ such that
\begin{align}\label{eq:quadapx4}
\ZZ_{\ChebSolve}
& \approx_{\epsilon / 3}
\left( \II - 
(1 - \theta) \WW^{1/2}_T\BB_{T,V(T)}\kh{\matlowtil \calDtil \matlowtil^\top}^{\dag}\BB_{T,V(T)}^\top\WW^{1/2}_T \right)^{-1},
\end{align}
whose cost is $O(\theta^{-0.5} \log(1/\eps))$
matrix-vector multiplications involving \\
$\II - 
(1 - \theta) \WW^{1/2}_T\BB_{T,V(T)}\kh{\matlowtil \calDtil \matlowtil^\top}^{\dag}\BB_{T,V(T)}^\top\WW^{1/2}_T$.
Here $\II$, $\WW^{1/2}_T$ and $\BB_{T,V(T)}$ can all be
applied in $O(\sizeof{T})$ time.
By Lemma~\ref{lem:partialChol},
$\kh{\matlowtil \calDtil \matlowtil^\top}^{\dag}$
can be applied in $O(\mathrm{nnz}(\SStil) + \sizeof{T}\theta^{-2}\eps^{-2}\log^3 n)$ time,
and $\partialChol(\SStil,\setof{v},\eps\theta/9)$ runs
in $O(\mathrm{nnz}(\SStil)\log^3 n + \sizeof{T}\theta^{-2}\eps^{-2}\log^5 n \operatorname{polyloglog}(n))$ time.

Inverting both sides of Equation~(\ref{eq:quadapx3}) and then combining it with Equation~(\ref{eq:quadapx4}) gives
\begin{align*}
\ZZ_{\ChebSolve}
& \approx_{\epsilon}
\kh{\II - 
(1 - \theta) \WW^{1/2}_T\BB_T\LL^\dag\BB_T^\top\WW^{1/2}_T}^{-1},
\end{align*}
so we can set $\widetilde{x} = \bb^\top \ZZ_{\ChebSolve} \bb$ and return it
as our overall estimate.
\end{proof}



Thus, the problem becomes efficiently approximating
Schur complements onto subsets of edges.
We give an algorithm $\quadApx$ that first computes
approximate Schur complements onto neighbors of
each vertex and then uses the algorithm $\QuadSolver$
in Lemma~\ref{lem:UsingApproxSchur}
to compute
\[
		\zz^\top \BB_{\Ev{v}}^\top \WW_{\Ev{v}}^{1 / 2}
		\kh{ \II - (1 - \theta)\WW_{\Ev{v}}^{1 / 2} \BB_{\Ev{v}} \LL^\dag \BB_{\Ev{v}}^\top\WW_{\Ev{v}}^{1 / 2} }^{-1}
		\WW_{\Ev{v}}^{1 / 2} \BB_{\Ev{v}} \zz
\]
for some vector $\zz$.
The pseudocode for $\quadApx$ is given in Algorithm~\ref{alg:normApprox}.
Note that in the pseudocode we use $G[C]$ to denote
$G$'s induced graph on a subset of vertices $C$, $\zz_C$ to
denote a $\sizeof{C}$-dimensional vector obtained
from $\zz$ by taking entries corresponding to vertices
in $C$,
and $\degu{v}$ to denote the number of edges incident
with $v$.
The performance of $\quadApx$ is characterized in Lemma~\ref{lem:normApprox}.
\begin{lemma}
	\label{lem:normApprox}
	Given a connected undirected graph $G = (V,E)$ with $n$ vertices,
	$m$ edges,
	positive edge weights
	$w : E \to \rea_{+}$, and associated Laplacian $\LL$,
	a set of vertices $V^Q \subset V$ such that
	$V = \setof{N(v) \,|\, v \in V^Q} \cup V^Q$,
	a vector $\zz\in \mathbb{R}^n$,
  	and scalars $0 < \theta \leq 1/2$, $0<\eps\leq1/2$,
  	the algorithm $\quadApx(G, \LL, V^Q, \zz, \theta, \eps\theta/9, \eps)$ returns a set of pairs
  	$\hat{N}^\Delta = \{ (v,\hat{n}_v^\Delta) \mid v \in V^Q \}$. With 	
  	high probability, the following statement holds:	For $\forall v \in V^Q$,
	\begin{align}\label{normApproxErrorBound}
		n_v^\Delta \approx_{\eps} \hat{n}_v^\Delta,
	\end{align}
	where
	\[
		n_v^\Delta = \zz^\top \BB_{\Ev{v}}^\top \WW_{\Ev{v}}^{1 / 2}
		\kh{ \II - (1 - \theta)\WW_{\Ev{v}}^{1 / 2} \BB_{\Ev{v}} \LL^\dag \BB_{\Ev{v}}^\top\WW_{\Ev{v}}^{1 / 2} }^{-1}
		\WW_{\Ev{v}}^{1 / 2} \BB_{\Ev{v}} \zz,
	\]
	and $\Ev{v} = \setof{(u,v)\,|\, u\sim v}$ is the set of edges incident with $v$.
	The total running time of this algorithm is bounded by
	$O(m(\theta^{-2}\eps^{-2}\log^8 n + \theta^{-2.5}\eps^{-2}\log^5 n \log(1/\eps)) \operatorname{polyloglog}(n)) $.
\end{lemma}

\begin{algorithm}
	\caption{$\quadApx(G,\SS, V^Q, \zz, \theta,\eps_1,\eps_2)$}
	\label{alg:normApprox}
	\Input{
		$G = (V,E)$: A graph. \\
		$\SS$: A graph Laplacian whose edges are supported on $V$. \\
		$V^Q \subset V$: a set of vertices \\
		$\zz \in \rea^{\sizeof{V}}$: a vector. \\
		$\theta$: a scalar between 0 and 1/2. \\
		$\eps_1$: the error parameter for Schur complement. \\
		$\eps_2$: the error parameter for $\QuadSolver$.
	}
	\Output{
		$\hat{N}^\Delta
		= \{ (v,\hat{n}_v^\Delta) \mid v \in V^Q \}$.
	}
	\If {$\sizeof{V^Q} = 1$}{
	 	Let $n$ and $m$ be the number of vertices and edges in $G$, respectively. \label{line:quadedgecase}\;
		Let $\BB$ be the $m\times n$ edge-vertex incidence matrix of $G$,
		and $\WW$ be the $m\times m$ diagonal edge weight matrix of $G$.\;
		Let $\hat{n}_v^\Delta = \QuadSolver(\BB,\WW,\WW^{1/2} \BB \zz,\theta, \eps_2, \SS)$ and return
		$\setof{(v, \hat{n}_v^\Delta)}$
			for the only vertex $v \in V^Q$. \label{line:quadedgecaseend}  \;
	}
	Let $\mathrm{vol} = \sum\nolimits_{v\in V^Q} \degu{v}$, and
		set $\eps_{\mathrm{schur}} = \eps_1 / \log_{\frac{4}{3}} \mathrm{vol}$. \label{line:epsschur}\;
	Let $V_4$ be vertices in $V_Q$ with $\degu{v} \geq \mathrm{vol} / 4$. \label{line:vfour} \;
	\If {$V_4 \neq \emptyset$}{
		\label{line:vfourstart}
		\For{each $v\in V_4$}{
			Let $C$ denote $v$ and its neighbors.\;
			$(\matlowtil, \calDtil, \SStil) \gets \partialChol(\SS, C, \eps_{\mathrm{schur}})  $\label{line:schur1} \;
			$\hat{N}^{\Delta (v)} \gets \quadApx(G[C],\SStil, \setof{v}, \zz_C, \theta, \eps_1 - \eps_{\mathrm{schur}}, \eps_2)$ \label{line:epsschur1} \;
		}
		Let $C$ denote vertices in $V^Q \setminus V_4$ and their neighbors. \label{line:twelve} \;
		$(\matlowtil, \calDtil, \SStil) \gets \partialChol(\SS, C, \eps_{\mathrm{schur}}) $ \label{line:schur2}\;
		\KwRet $\quadApx(G[C],\SStil', V^Q \setminus V_4,
		\zz_{C}, \theta, \eps_1 - \eps_{\mathrm{schur}},\eps_2) \cup \kh{ \bigcup\nolimits_{v\in V_4} \hat{N}^{\Delta (v)} }$ \label{line:epsschur2}
	}
	Divide $V^Q$ into two parts $V^{(1)}$ and $V^{(2)}$ such that both $\sum\nolimits_{v\in V^{(1)}} \degu{v}$ and
	$\sum\nolimits_{v\in V^{(2)}} \degu{v}$ are in the range
		$\left[\frac{1}{4}\mathrm{vol},\frac{3}{4}\mathrm{vol}\right]$. \label{line:fifteen} \;
	\For{$i = 1$ to $2$}{
		Let $C$ denote vertices in $V^{(i)}$ and their neighbors. \;
		$(\matlowtil, \calDtil, \SStil) \gets \partialChol(\SS, C, \eps_{\mathrm{schur}}) $ \label{line:schur3} \;
		$\hat{N}^{\Delta (i)} \gets \quadApx(G[C],\SStil, V^{(i)}, \zz_{C}, \theta, \eps_1 - \eps_{\mathrm{schur}} ,\eps_2)$ \label{line:epsschur3}\;
	}
	\KwRet $\hat{N}^{\Delta (1)} \cup \hat{N}^{\Delta (2)}$ \label{line:return}.
\end{algorithm}

\begin{proof}[Proof of Lemma~\ref{lem:normApprox}] 


Let $\mathrm{volume}(V^Q)$ denote the quantity $\mathrm{vol}$ on Line~\ref{line:epsschur}.
We first observe that every time we recursively call $\quadApx$,
one of the following two events occurs:
\begin{enumerate}
\item $\mathrm{volume}(V^Q)$ becomes no more than its $3/4$ (Lines~\ref{line:epsschur2} and~\ref{line:epsschur3}), or
\item $\sizeof{V^Q}$ becomes $1$ (Lines~\ref{line:epsschur1}).
\end{enumerate}
When $\sizeof{V^Q} = 1$, the algorithm
will go to Lines~\ref{line:quadedgecase}\,-\,\ref{line:quadedgecaseend},
and hence the recursion depth is only $1$.
Then,
as we set $V^Q = V$ in the earliest call to $\quadApx$,
we have that the total recursion depth is no more
than $\log_{\frac{4}{3}} \mathrm{volume}(V) = \log_{\frac{4}{3}} 2m$.

We then give guarantees for our approximations. Note that we set $\eps_{\mathrm{schur}} = \eps_1 / \log_{\frac{4}{3}} \mathrm{vol}$ (Line~\ref{line:epsschur}),
and when recursively calling $\quadApx$ we set the $\eps_1$ of the recursive call to
$\eps_1 - \eps_{\mathrm{schur}}$
 (Line~\ref{line:epsschur1},~\ref{line:epsschur2} and~\ref{line:epsschur3}).
Then, since we set $\eps_1 = \eps\theta/9$ in the earliest call to $\quadApx$, we have that $\eps_{\mathrm{schur}} \leq \frac{\eps\theta}{9} /\log_{\frac{4}{3}} 2m$
always holds. Coupled with the fact that the total recursion depth is no more than $\log_{\frac{4}{3}} 2m$,
on Line~\ref{line:quadedgecaseend} we have that
\begin{align*}
	\SS \approx_{\eps\theta/9} \SC(\LL^G, V(\Ev{v}))
\end{align*}
holds for the only vertex $v\in V^Q$, where
$\LL^G$ is the Laplacian matrix of the graph in
the earliest call to $\quadApx$ (i.e., the original graph). Then by
Lemma~\ref{lem:UsingApproxSchur}, $\hat{n}_v^\Delta$ on
Line~\ref{line:quadedgecaseend} satisfies
\begin{align*}
\hat{n}_v^\Delta
\approx_{\epsilon}
\zz^\top \BB_{\Ev{v}}^\top \WW_{\Ev{v}}^{1 / 2}
		\kh{ \II - (1 - \theta)\WW_{\Ev{v}}^{1 / 2} \BB_{\Ev{v}} \LL^\dag \BB_{\Ev{v}}^\top\WW_{\Ev{v}}^{1 / 2} }^{-1}
		\WW_{\Ev{v}}^{1 / 2} \BB_{\Ev{v}} \zz.
\end{align*}

We now analyze the running time. 

  	Let $T(m)$ denote the running time of
  	$\quadApx(G,\SS, V^Q, \zz, \theta,\eps_1,\eps_2)$,
  	where $m = \mathrm{volume}(V^Q)$.
  	Let $\ncur$ and $\mcur$ denote the number of vertices and the number of edges
	in $\LL$ in the current call, respectively.
  	We first assume $\QuadSolver$ to be an $O(1)$ operation,
  	and hence $T(m) = O(1)$ for $\sizeof{V^Q} = 1$.
  	For $\sizeof{V^Q} > 1$, We consider the set $V_4$ on Line~\ref{line:vfour}:
  	\begin{enumerate}
	\item If $V_4$ is not empty, the algorithm goes to Lines~\ref{line:vfourstart}\,-\,\ref{line:epsschur2}.
	Since there are at most $4$ vertices in $V_4$, and
	by our assumption the recursive calls to $\quadApx$ on
	Line~\ref{line:epsschur1} all run in $O(1)$ time,
	we have by Lemma~\ref{lem:partialChol} Lines~\ref{line:vfourstart}\,-\,\ref{line:schur2} runs in
	total $O((\mcur\log^3 n + \ncur\eps_{\mathrm{schur}}^{-2}\log^5 n)\operatorname{polyloglog}(n))$ time.
	Since $V_4$ is not empty, we have $\mathrm{volume}(V\setminus V_4) \leq \frac{3}{4}\mathrm{volume}(V)$. Hence,
	the running time of the recursive call to $\quadApx$ on Line~\ref{line:epsschur2} is at most $T(3m/4)$.
	\item If $V_4$ is empty, the algorithm goes to Lines~\ref{line:fifteen}\,-\,\ref{line:return}. By Lemma~\ref{lem:partialChol}, the calls to\\ $\partialChol$ on Line~\ref{line:schur3} run in total $O((\mcur\log^3 n + \ncur\eps_{\mathrm{schur}}^{-2}\log^5 n)\operatorname{polyloglog}(n))$ time. The running time of the recursive calls to $\quadApx$ on Line~\ref{line:epsschur3} is $T(m_1) + T(m - m_1)$, where $m_1 = \mathrm{volume}(V^{(1)})$ is in the range $\left[m/4, 3m/4 \right]$.
\end{enumerate}

Since $\eps_{\mathrm{schur}} = O(\theta\eps/\log n)$, $\ncur = O(m)$,` and by Lemma~\ref{lem:partialChol}
$\mcur = O(\ncur \eps_{\mathrm{schur}}^{-2} \log n)
= O(m \theta^{-2} \eps^{-2} \log^3 n)$, we have
in the worst case
\begin{align*}
	T(m) = T(3m/4) + T(m/4) + O(m\theta^{-2} \eps^{-2}\log^7 n \operatorname{polyloglog}(n)),
\end{align*}
which gives $T(m) = O(m\theta^{-2} \eps^{-2}\log^8 n \operatorname{polyloglog}(n))$.

Note that we get this running time under the assumption that
$\QuadSolver$ is an $O(1)$ operation. Thus, we also need to analyze the total running time of the calls to $\QuadSolver$ on Line~\ref{line:quadedgecaseend}.

By Lemma~\ref{lem:UsingApproxSchur},
the $\QuadSolver(\BB,\WW,\WW^{1/2} \BB \zz,\theta, \eps_2, \SS)$ on Line~\ref{line:quadedgecaseend}
runs in\\ $O(\mathrm{nnz}(\SS)\theta^{-0.5}\log^3 n \log(1/\eps) + \degu{v}\theta^{-2.5}\eps^{-2}\log^5 n \log(1/\eps) \operatorname{polyloglog}(n))$ time, where $v$ indicates
the only vertex in $V^Q$ and $\degu{v}$ is the number of edges incident to $v$.
By Lemma~\ref{lem:partialChol},
we have $\mathrm{nnz}(\SS) = O(\degu{v}\eps_{\mathrm{schur}}^{-2}\log n) = O(\degu{v}\theta^{-2}\eps^{-2}\log^3 n)$.
Then, summing this running time over all vertices gives
$O(m\theta^{-2.5}\eps^{-2}\log^5 n \log(1/\eps) \operatorname{polyloglog}(n))$, which plus $T(m)$
gives the overall running time of this algorithm.

\end{proof}


\subsection{Approximating $\mathcal{C}_\theta^\Delta(v)$}

We give the pseudocode of the algorithm $\VCComp$ which
approximates $\theta$-Kirchhoff vertex centrality $\mathcal{C}_\theta^\Delta(v)$ for all $v\in V$
in Algorithm~\ref{alg:VCComp}.
Note that in this algorithm we once again
invoke the Laplacian solver 
of~\cite{CKMPPRX14}.
The performance of $\VCComp$ is characterized in
Theorem~\ref{lem:vertexbyjlschur}.
Analyzing this algorithm gives the main result for estimating vertex
centralities.


\begin{algorithm}[h]
	\caption{$\VCComp(G=(V,E), w, \theta,\epsilon)$}
	\label{alg:VCComp}
	\Input{
		$G = (V,E)$, $w$: A connected undirected graph with positive edges \\weights $w: E\to \rea_{+}$. \\
		$\theta$: A scalar between 0 and 1/2. \\
		$\epsilon$: Error of the centrality estimate per vertex. \\
	}
	\Output{
		$\hat{C}^\Delta = \left\{ (v, \hat{c}_v^\Delta) \mid v \in V \right\}$.
	}
	Let $\zz_1,\ldots,\zz_M$ be independent random $\pm 1$ vectors, where $M = \ceil{432\eps^{-2}\ln(2n)}$. \;
	\For{$i = 1$ to $M$}{
		$\yy_i \gets \LaplSolver(\LL^G, \zz_i, \frac{1}{36}\theta\eps n^{-7} U^{-4})$ \;
		$\kh{ \hat{N}^{\Delta (i)} = \{ (v,\hat{n}_v^{\Delta (i)}) \mid v \in V \} } \gets
			\quadApx(\LL^G, V, \yy_i, \theta,\theta\eps/27, \epsilon/3)$
	}
	For each $v\in V$ compute $\hat{c}_v^\Delta = (1 - \theta) \frac{n}{M} \sum\limits_{i=1}^M \hat{n}_v^{\Delta (i)}$ and
		return $\hat{C}^\Delta = \left\{ (v, \hat{c}_v^\Delta) \mid v \in V \right\}$.
\end{algorithm}

\begin{proof}[Proof of Theorem~\ref{lem:vertexbyjlschur}]
The running time is the total cost of $O(\eps^{-2}\log n)$
calls to $\LaplSolver$ each of which runs
in $O(m\log^{1.5} n\log(\frac{1}{\eps\theta})\operatorname{polyloglog}(n))$ time,
and $O(\eps^{-2}\log n)$ calls to
$\quadApx$ each of which runs in
$O(m(\theta^{-2}\eps^{-2}\log^8 n + \theta^{-2.5}\eps^{-2}\log^5 n \log(1/\eps)) \operatorname{polyloglog}(n))$
time.

In the rest of this proof, we will use the matrix $\CC_v$,
defined as
\begin{align*}
\CC_v \defeq \BB_{\Ev{v}}^\top \WW_{\Ev{v}}^{1 / 2}
		\kh{ \II - (1 - \theta)\WW_{\Ev{v}}^{1 / 2} \BB_{\Ev{v}} \LL^\dag \BB_{\Ev{v}}^\top\WW_{\Ev{v}}^{1 / 2} }^{-1}
		\WW_{\Ev{v}}^{1 / 2} \BB_{\Ev{v}},
\end{align*}
to simplify notation.

Since $M = \ceil{432\eps^{-2}\ln(2n)} \geq 48 \kh{\eps/3}^{-2} \ln(2n)$, by Lemma~\ref{lem:mcL}, we have
\begin{align}\label{eq:ver1}
	\frac{1}{M}\sum\limits_{i=1}^M
	\zz_i^\top \LL^\dag \CC_v \LL^\dag\zz_i
		\approx_{\eps/3}
	\trace{\LL^\dag \CC_v \LL^\dag}.
\end{align}
By Lemma~\ref{lem:Lower}, we have
\begin{align*}
	\trace{\LL^\dag \CC_v \LL^\dag}
		\geq \frac{2}{n^2U^2},
\end{align*}
and hence
\begin{align}\label{eq:quadlower}
	\frac{1}{M}\sum\limits_{i=1}^M
	\zz_i^\top \LL^\dag \CC_v \LL^\dag\zz_i
		\geq
		\exp(-\eps/3)\frac{2}{n^2U^2} \geq \frac{1}{n^2U^2},
\end{align}
where the second inequality follows by $0 < \eps \leq 1/2$.

Since we set $\delta = \frac{1}{36}\theta\eps n^{-7} U^{-4}$ when invoking
	$\LaplSolver$, by Lemma~\ref{lem:laplsolver},
	\begin{align*}
		\len{\yy_i - \LL^\dag \zz_i}_{\LL} \leq
		\frac{1}{36}\theta\eps n^{-7} U^{-4}
		\len{\LL^\dag \zz_i}_{\LL},
	\end{align*}
	holds for each $i$.
	Then, by Lemma~\ref{lem:SolverError},
	we have that
	\begin{align*}
		\left|
		\yy_i^\top \CC_v \yy_i -
		\zz_i^\top \LL^\dag \CC_v \LL^\dag \zz_i
		\right|
		\leq \frac{1}{6}\eps n^{-2} U^{-2}
	\end{align*}
	holds for each $i$. We then have
	\begin{align*}
		&\left|
		\frac{1}{M} \sum\limits_{i=1}^M \yy_i^\top \CC_v \yy_i -
		\frac{1}{M} \sum\limits_{i=1}^M \zz_i^\top \LL^\dag \CC_v \LL^\dag \zz_i \right| \\
		\leq &
		\frac{1}{M} \sum\limits_{i=1}^M
		\left| \yy_i^\top \CC_v \yy_i
		- \zz_i^\top \LL^\dag \CC_v \LL^\dag \zz_i \right| \\
		\leq & \frac{1}{6}\eps n^{-2} U^{-2} \\
		\leq  &\frac{1}{6}\eps \kh{
		\frac{1}{M} \sum\limits_{i=1}^M \zz_i^\top \LL^\dag \CC_v \LL^\dag \zz_i },
	\end{align*}
	where the last inequality follows by~(\ref{eq:quadlower}).
	Thus,
	\begin{align*}
		(1 - \eps/6)
		\frac{1}{M} \sum\limits_{i=1}^M \zz_i^\top \LL^\dag \CC_v \LL^\dag \zz_i
		\leq
		\frac{1}{M}\sum\limits_{i=1}^M \yy_i^\top \CC_v \yy_i
		\leq
		(1 + \eps/6)
		\frac{1}{M} \sum\limits_{i=1}^M \zz_i^\top \LL^\dag \CC_v \LL^\dag \zz_i,
	\end{align*}
	which implies
	\begin{align}\label{eq:ver2}
		\frac{1}{M} \sum\limits_{i=1}^M \zz_i^\top \LL^\dag \CC_v \LL^\dag \zz_i \approx_{\eps/3}
		\frac{1}{M}\sum\limits_{i=1}^M \yy_i^\top \CC_v \yy_i.
	\end{align}
	
	By Lemma~\ref{lem:normApprox}, we have
	\begin{align}\label{eq:ver3}
		\hat{n}_v^{\Delta (i)} \approx_{\eps/3}
		\yy_i^\top \CC_v \yy_i.
	\end{align}
	Combining Equation~(\ref{eq:ver1}),~(\ref{eq:ver2}) and~(\ref{eq:ver3}), we have
	\begin{align*}
		\frac{1}{M}\sum\limits_{i=1}^M \hat{n}_v^{\Delta (i)}
		\approx_{\eps}
		\trace{\LL^\dag \CC_v \LL^\dag},
	\end{align*}
	which coupled with the fact that
	\begin{align*}
		\mathcal{C}_{\theta}^{\Delta}(v)
		&= \Kf{G\bsk \Ev{v}} - \Kf{G} \qquad \text{by definition} \\
		&= n\kh{\trace{\kh{\LL\bsk \Ev{v}}^\dag} - \trace{\LL^\dag}} \qquad \text{by Fact~\ref{fact:kirchtr}} \\
		&= n (1-\theta) \trace{\LL^\dag \CC_v \LL^\dag}
		\qquad \text{by Equation~(\ref{eq:woodbury})}
	\end{align*}
	gives the guarantee of our approximation.
\end{proof}

%% file: conclusion.tex
\section{Conclusion and Future Work}\label{sec:conclude}

The Kirchhoff index arises in many applications
such as noisy consensus problems~\cite{PaBa14} and
social recommender systems~\cite{WoLiCh16}.
It is a global index, and any changes of network structure,
e.g.  weight of edges,  can be reflected in this popular index.
In this paper, we proposed to use  Kirchhoff index
 as a global metric of the importance of edges in
 weighted undirected networks.
For any network, when the weight of any edge $e$ is changed from $w(e)$
from $\theta w(e)$, the  Kirchhoff index of the resulting graph
will strictly increase, with the increase
deciphering the importance of edge $e$.
We used the   Kirchhoff index of the new graph,
or its increment with respect to the original graph,
as the  $\theta$-Kirchhoff edge centrality. 
We demonstrated experimentally that this new global measure of
 edge centrality has a more discriminating power
than edge betweenness, spanning edge centrality,
and current-flow centrality.

However, the time cost of exactly computing the
$\theta$-Kirchhoff edge centrality is prohibitive.
To overcome this weakness, we introduced two approaches that estimate
the $\theta$-Kirchhoff edge centrality for all edges in nearly linear time.
Our proposed centrality metrics are the first global measure of centrality
that can be estimated in nearly linear time.
Our algorithms combine techniques from several recent works
on graph algorithms~\cite{LeSiWo15, DKP+17}.
We also extend these ideas to develop efficient algorithms for estimating
$\theta$-Kirchhoff vertex centrality, as well as estimating the
Kirchhoff edge centrality to a set of edges.
This raises the possibility of designing highly efficient
algorithms that can detect the set of $k$ most
influential edges, that is, the $k$ edges whose $\theta$-deletion leads to
the largest increase of the Kirchhoff index.

Despite the advantages of our algorithms, their theoretical performance still has much room for improvement
, both in the overhead of logarithmic factors and
the dependency on $\theta$.
The latter is particularly interesting because our two algorithms
for estimating edge centrality can perform better under different
regimes of $\theta$.
On the other hand, the importance of centrality measures in graph
mining means it is just as, if not more, interesting to study the
practical behaviors of our algorithms.
Specifically, to see if they are reasonably fast and accurate on
massive networks with millions of vertices and edges.
Recent packages for solving large scale linear systems and related
tasks~\cite{LAMG,CMG,NetworkKIT,LaplacianJL} should greatly facilitate
such a study.
Moreover, the significantly higher deviations from our experiments
suggest the question of whether it is possible to theoretically model
the advantages/disadvantages of the many centrality measures.

Finally, it should be mentioned that as an application of the introduced edge centrality, we studied the vertex centrality based on the idea of the definition for $\mathcal{C}_\theta^\Delta(e)$. Actually, we can also define the centrality of a vertex $v$ as the Kirchhoff index of the graph $G\bsk \Ev{v}$, the algorithm for the $\eps$-approximation of which is similar to $\ECComp$. We thus omit the algorithmic details of this version of vertex centrality for the lack of space. Another reason for ignoring this algorithm is that our main focus is the edge centrality. 

%% file: centrality.bbl
\newcommand{\etalchar}[1]{$^{#1}$}

%% file: Errors.tex
\section{Proofs of Our Version of Sherman-morrision and Woodbury Formulas}\label{sec:morrison}

In this section, we give detailed proofs for the Sherman-Morrision and Woodbury
formulas we used, i.e., Equations~(\ref{eq:morrison}) and~(\ref{eq:woodbury}).

In the proofs, we will use the matrix $\PPi$
defined as
\begin{align*}
	\PPi \defeq \LL\LL^\dag = \II - \frac{1}{n}\vecone\vecone^{1},
\end{align*}
where $\vecone$ is the vector with all entries being $1$.

\begin{proof}[Proof of Equation~(\ref{eq:morrison})]
	First, we have
	\begin{align}\label{eq:idmorrison}
		\bb_e\kh{1 - (1-\theta)w(e)\bb_e^\top\LL^\dag\bb_e}
		= \bb_e - (1 - \theta)w(e)\bb_e\bb_e^\top\LL^\dag\bb_e
		= \kh{\LL - (1 - \theta)w(e)\bb_e\bb_e^\top}\LL^\dag\bb_e,
	\end{align}
	where the second equality follows by $\bb_e = \PPi\bb_e = \LL\LL^\dag \bb_e$.
	
	Since $\theta < 1$,
	we have that
	$\LL - (1 - \theta)w(e)\bb_e\bb_e^\top$ is a Laplacian matrix and
	$1 - (1 - \theta)w(e)\bb_e^\top\LL^\dag\bb_e$
	is strictly positive.
	Thus, Equation~(\ref{eq:idmorrison})
	implies
	\begin{align*}
		\kh{\LL - (1 - \theta)w(e)\bb_e\bb_e^\top}^\dag \bb_e
		=
		\frac{\LL^\dag \bb_e}{1 - (1 - \theta)w(e)\bb_e^\top\LL^\dag\bb_e}.
	\end{align*}
	Then, we have
	\begin{align*}
		\LL^\dag &= \kh{\LL - (1 - \theta)w(e)\bb_e\bb_e^\top}^\dag \kh{\LL - (1 - \theta)w(e)\bb_e\bb_e^\top}\LL^\dag \\
		&= \kh{\LL - (1 - \theta)w(e)\bb_e\bb_e^\top}^\dag
			\kh{\PPi - (1 - \theta)w(e)\bb_e\bb_e^\top\LL^\dag} \\
		&= \kh{\LL - (1 - \theta)w(e)\bb_e\bb_e^\top}^\dag -
			(1 - \theta)w(e)\kh{\LL - (1 - \theta)w(e)\bb_e\bb_e^\top}^\dag
			\bb_e\bb_e^\top\LL^\dag \\
		&= \kh{\LL - (1 - \theta)w(e)\bb_e\bb_e^\top}^\dag -
			(1 - \theta)\frac{w(e)\LL^\dag\bb_e\bb_e^\top\LL^\dag}{1 - (1 - \theta)w(e)\bb_e^\top\LL^\dag\bb_e},
	\end{align*}
	which implies Equation~(\ref{eq:morrison}).
\end{proof}

\begin{proof}[Proof of Equation~(\ref{eq:woodbury})]
	First, we have
	\begin{align}\label{eq:idwoodbury}
		&\BB_{T}^\top \WW_{T}^{1 / 2}
\kh{ \II - (1 - \theta)\WW_{T}^{1 / 2} \BB_{T}
  \LL^\dag
\BB_{T}^\top\WW_{T}^{1 / 2} } \notag\\
		=& \BB_{T}^\top \WW_{T}^{1 / 2} -
			(1 - \theta)\BB_{T}^\top \WW_{T}\BB_{T}
  \LL^\dag
\BB_{T}^\top\WW_{T}^{1 / 2} \notag\\
		=& \kh{\LL - (1-\theta)\BB_T^\top\WW_T\BB_T}\LL^\dag
\BB_{T}^\top\WW_{T}^{1 / 2},
	\end{align}
	where the second equality follows by
	$
		\BB_{T}^\top
		= \PPi \BB_{T}^\top
		= \LL \LL^\dag \BB_T^\top
	$.
	
	Since $\theta < 1$, we have that
	$\II - (1 - \theta)\WW_{T}^{1 / 2} \BB_{T}
  \LL^\dag
\BB_{T}^\top\WW_{T}^{1 / 2}$
	is positive definite
	and\\
	$\LL - (1-\theta)\BB_T^\top\WW_T\BB_T$
	is a Laplacian matrix.
	Thus, Equation~(\ref{eq:idwoodbury}) implies
	\begin{align*}
		\kh{\LL - (1-\theta)\BB_T^\top\WW_T\BB_T}^\dag \BB_{T}^\top \WW_{T}^{1 / 2}
		=
		\LL^\dag \BB_{T}^\top\WW_{T}^{1 / 2}
		\kh{\II - (1 - \theta)\WW_{T}^{1 / 2} \BB_{T}
  \LL^\dag
\BB_{T}^\top\WW_{T}^{1 / 2}}^{-1}.
	\end{align*}
	Then, we can write $\LL^\dag$ as
	\begin{align*}
		&
			\kh{\LL - (1-\theta)\BB_T^\top\WW_T\BB_T}^\dag
			\kh{\LL - (1-\theta)\BB_T^\top\WW_T\BB_T}
			\LL^\dag \\
		= & \kh{\LL - (1-\theta)\BB_T^\top\WW_T\BB_T}^\dag
			\kh{\PPi - (1-\theta)\BB_T^\top\WW_T\BB_T\LL^\dag}\\
		= & \kh{\LL - (1-\theta)\BB_T^\top\WW_T\BB_T}^\dag
			-
			(1 - \theta)
			\kh{\LL - (1-\theta)\BB_T^\top\WW_T\BB_T}^\dag
			\BB_T^\top\WW_T\BB_T\LL^\dag \\
		= & \kh{\LL - (1-\theta)\BB_T^\top\WW_T\BB_T}^\dag
			-
			(1 - \theta)
			\LL^\dag \BB_{T}^\top\WW_{T}^{1 / 2}
		\kh{\II - (1 - \theta)\WW_{T}^{1 / 2} \BB_{T}
  \LL^\dag
\BB_{T}^\top\WW_{T}^{1 / 2}}^{-1}
			\WW_T^{1/2}\BB_T\LL^\dag,
	\end{align*}
	which implies Equation~(\ref{eq:woodbury}).
\end{proof}

\section{Approximations When Subtracted From Identity Matrix}\label{sec:Errors}

In this section, we bound the transfer of approximations between $\AA$ and $\BB$
to approximations between $\II - (1 - \theta) \AA$ and $\II - (1 - \theta) \BB$.

\begin{proof}[Proof of Lemma~\ref{lem:SubtractError}]
The given condition with the approximation
can be written as:
\[
\left(1 - 2\epsilon \right) \AA
\preceq \BB \preceq
\left(1 + 2\epsilon \right) \AA,
\]
which implies
\[
\II - (1 - \theta) \left( 1 + 2\epsilon \right) \AA
\preceq \II - (1 - \theta) \BB \preceq
\II - (1 - \theta) \left(1 - 2\epsilon \right) \AA.
\]
Since $0 \preceq \AA \preceq \II$, we have the following
lower bound:
\begin{align*}
\II - (1 - \theta) \left( 1 + 2\epsilon \right) \AA
& = \left( 1 - \frac{2 \epsilon}{\theta} \right)
\left( \II - \left( 1 - \theta \right) \AA \right)
+ \frac{2 \epsilon}{\theta} \II - \left( 1 - \theta \right)
\left( 2 \epsilon + \frac{2 \epsilon}{\theta} \right) \AA\\
& \succeq \left( 1 - \frac{2 \epsilon}{\theta} \right)
\left( \II - \left( 1 - \theta \right) \AA \right)
+ \left[\frac{2 \epsilon}{\theta} - \left( 1 - \theta \right)
\left( 2 \epsilon + \frac{2 \epsilon}{\theta} \right) \right] \AA.
\end{align*}
The coefficient on the trailing $\AA$ in turn simplifies to
$2 \epsilon - (1 - \theta) 2 \epsilon \geq 0$.

Similarly for the upper bound we get:
\begin{align*}
\II - (1 - \theta) \left( 1 - 2\epsilon \right) \AA
& = \left( 1 + \frac{2 \epsilon}{\theta} \right)
\left( \II - \left( 1 - \theta \right) \AA \right)
- \frac{2 \epsilon}{\theta} \II + \left( 1 - \theta \right)
\left( 2 \epsilon + \frac{2 \epsilon}{\theta} \right) \AA\\
& \preceq \left( 1 + \frac{2 \epsilon}{\theta} \right)
\left( \II - \left( 1 - \theta \right) \AA \right).
\end{align*}
Then the final bound involving $\exp(3 \epsilon / \theta)$
follows from the condition of $\epsilon / \theta$ being small.
\end{proof}

\section{Error Tracking for Laplacian Solvers}

In this section, we provide more details on error tracking
for Laplacian solvers in Section~\ref{sec:edgebyjl} and~\ref{sec:vertexbyjlschur} in a way similar to Section 4 of~\cite{SS11}.

We first give bounds on eigenvalues of $\LL$.
Let $\LL$ be the Laplacian matrix of a graph $G = (V,E)$ with $n$ vertices, $m$ edges and edge weights all in the range $[1,U]$.
Let $0 = \lambda_1 \leq \lambda _2 \leq \ldots \leq \lambda_n$
be the eigenvalues of $\LL$,
and $0 = \nu_1 \leq \nu_2 \leq \ldots \leq \nu_n$ be
the eigenvalues of the normalized Laplacian matrix,
$\NN \defeq \DD^{-1/2} \LL \DD^{-1/2}$,
of $G$.
It is easy to verify that
$\nu_i \leq \lambda_i \leq nU \nu_i$
holds for all $i$.
Let
\begin{align*}
	\phi_G = \min\limits_{S \subset V}
	\frac{\left| \partial(S) \right|}{\min\kh{d(S),d(V\setminus S)}}
\end{align*}
be the conductance of $G$,
where $\left| \partial(S) \right|$ denotes the total weights
of edges with one endpoint in $S$ and the other endpoint
in $V\setminus S$, and $d(S)$ denotes the total degree
of vertices in $S$.
Then, we can 
bound $\lambda_2$ by
\begin{align}\label{lambda2}
	\lambda_2 & \geq \nu_2
	\geq \phi_{G}^2/2 \qquad \text{by Cheeger's inequality} \notag \\
	&\geq \kh{\frac{1}{n^2 U}}^2 / 2 \qquad \text{since all edge weights are in $[1,U]$} \notag \\
	&= \frac{1}{2n^4 U^2}.
\end{align}
We then bound $\lambda_n$ using the fact
that $\LL^G \preceq U\LL^{K_n}$, where
$K_n$ is the complete graph of $n$ vertices.
Thus,
\begin{align}\label{lambdan}
	\lambda_n^G \leq \lambda_n^{K_n}U = nU.
\end{align}
From~(\ref{lambda2}) and~(\ref{lambdan}) it is immediate
that
\[
	\frac{1}{2n^4 U^2} \PPi \preceq \LL \preceq nU \II
	\qquad \text{and} \qquad
	\frac{1}{n U} \PPi \preceq \LL^\dag \preceq 2n^4 U^2 \II
\]
hold,
where $\PPi \defeq \LL \LL^\dag = \II - \frac{1}{n}\vecone\vecone^\top$.

We will also need to use the inequality
\begin{align*}
	\left| x^2 - y^2 \right| \leq (2|y| + |x-y|)|x-y|
\end{align*}
for scalars $x,y$, which follows by
\begin{align*}
	\left| x^2 - y^2 \right|
	\leq (|x| + |y|)|x - y| \leq (|y| + |y+(x-y)|)|x-y|
	\leq (2|y|+|x-y|)|x-y|.
\end{align*}

\subsection{Error Tracking for the Laplacian Solver in Section~\ref{sec:edgebyjl}}
\label{sec:Errorsedge}

\begin{proof}[Proof of Lemma~\ref{lem:edgeSolveError}]
	The lhs of inequality~(\ref{eq:edgeSolveError}) can be
	written as
	\begin{align*}
		\left| \norm{\yy}_{\bb_e\bb_e^\top}^2 -
		\norm{\LL^\dag\zz}_{\bb_e\bb_e^\top}^2
		 \right|.
	\end{align*}
	We first bound the value
	$\left| 
		\norm{\yy}_{\bb_e\bb_e^\top} -
		\norm{\LL^\dag\zz}_{\bb_e\bb_e^\top} \right|$
	by
	\begin{align*}
		\left| 
		\norm{\yy}_{\bb_e\bb_e^\top} -
		\norm{\LL^\dag\zz}_{\bb_e\bb_e^\top} \right|
		& \leq
			\len{\yy - \LL^\dag \zz}_{\bb_e\bb_e^\top}
			\qquad \text{by the triangle inequality of norms}
			\\
		& \leq \len{\yy - \LL^\dag \zz}_{\LL}
		\qquad \text{since $\bb_e\bb_e^\top \preceq \LL$} \\
		& \leq \delta \len{\LL^\dag \zz}_{\LL}
			= \delta \sqrt{\zz^\top \LL^\dag \zz} \\
		& \leq \delta n^{0.5}
			\sqrt{\frac{\zz^\top \LL^\dag \zz}{\zz^\top \zz}}
			\qquad \text{since $\len{\zz}^2 \leq n$} \\
		& \leq \sqrt{2}\delta n^{2.5} U
			\qquad \text{since $\LL^\dag \leq 2n^4 U^2 \II$.}
	\end{align*}
	We then use the inequality
	$\left| x^2 - y^2 \right| \leq (2|y| + |x-y|)|x-y|$
	to bound $\left| \norm{\yy}_{\bb_e\bb_e^\top}^2 -
		\norm{\LL^\dag\zz}_{\bb_e\bb_e^\top}^2
		 \right|$:
	\begin{align*}
		&\left| \norm{\yy}_{\bb_e\bb_e^\top}^2 -
		\norm{\LL^\dag\zz}_{\bb_e\bb_e^\top}^2
		 \right| \\
		 \leq &
		 \kh{ 2\norm{\LL^\dag \zz}_{\bb_e\bb_e^\top} +
		 \left|
		\norm{\yy}_{\bb_e\bb_e^\top} -
		\norm{\LL^\dag\zz}_{\bb_e\bb_e^\top} \right| }
		\left|
		\norm{\yy}_{\bb_e\bb_e^\top} -
		\norm{\LL^\dag\zz}_{\bb_e\bb_e^\top} \right| \\
		\leq &
		\kh{ 2\norm{\LL^\dag \zz}_{\LL} +
		\sqrt{2}\delta n^{2.5} U
		 }
		\sqrt{2}\delta n^{2.5} U
		\qquad \text{by $\bb_e\bb_e^\top \preceq \LL$
			and the above bound} \\
		\leq
		&\kh{ 2\sqrt{2} n^{2.5} U +
		\sqrt{2}\delta n^{2.5} U
		 }
		\sqrt{2}\delta n^{2.5} U \qquad
		\text{since $\len{\zz}^2 \leq n$
		and $\LL^\dag \leq 2n^4 U^2 \II$} \\
		\leq & 6\delta n^5 U^2 \qquad
		\text{by $\delta < 1$.}
	\end{align*}
\end{proof}

\begin{proof}[Proof of Lemma~\ref{lem:edgeLower}]
	\begin{align*}
		\trace{\LL^\dag \bb_e \bb_e^\top \LL^\dag}
		&= \bb_e^\top \LL^\dag \LL^\dag \bb_e
		\qquad \text{by cyclicness of trace} \\
		&= 2
			\frac{\bb_e^\top \kh{\LL^\dag}^2 \bb_e}{\bb_e^\top\bb_e}
			\qquad \text{since $\len{\bb_e}^2 = 2$} \\
		&\geq \frac{2}{n^2 U^2} \qquad
			\text{since $\kh{\LL^\dag}^2\succeq \frac{1}{n^2 U^2} \PPi$.}
	\end{align*}
\end{proof}

\subsection{Error Tracking for the Laplacian Solver in Section~\ref{sec:vertexbyjlschur}}
\label{sec:Errorsvertex}

\begin{proof}[Proof of Lemma~\ref{lem:SolverError}]
	The lhs of inequality~(\ref{eq:SolverError}) can be seen
	as the difference between the following two values:
	\begin{align*}
		\norm{\yy}_{\BB_{T}^\top \WW_{T}^{1 / 2}
\kh{ \II - (1 - \theta)\WW_{T}^{1 / 2} \BB_{T}\LL^\dag
\BB_{T}^\top\WW_{T}^{1 / 2} }^{-1}
\WW_{T}^{1 / 2} \BB_{T}}^2 \\
		\norm{\LL^\dag\zz}_{\BB_{T}^\top \WW_{T}^{1 / 2}
\kh{ \II - (1 - \theta)\WW_{T}^{1 / 2} \BB_{T}\LL^\dag
\BB_{T}^\top\WW_{T}^{1 / 2} }^{-1}
\WW_{T}^{1 / 2} \BB_{T}}^2.
	\end{align*}
	By the triangle inequality of norms, the difference
	between the square roots of these two values is at most
	\begin{align*}
		&\len{\yy - \LL^\dag \zz}_{\BB_{T}^\top \WW_{T}^{1 / 2}
\kh{ \II - (1 - \theta)\WW_{T}^{1 / 2} \BB_{T}\LL^\dag
\BB_{T}^\top\WW_{T}^{1 / 2} }^{-1}
\WW_{T}^{1 / 2} \BB_{T}}
			\\
		\leq &
		\theta^{-0.5}\len{\yy - \LL^\dag \zz}_{\BB_{T}^\top \WW_{T}\BB_{T}}
		\qquad \text{since $\kh{ \II - (1 - \theta)\WW_{T}^{1 / 2} \BB_{T}\LL^\dag \BB_{T}^\top\WW_{T}^{1 / 2} }^{-1} \preceq \frac{1}{\theta} \II$} \\
		\leq & \theta^{-0.5}\len{\yy - \LL^\dag \zz}_{\LL}
		\qquad \text{since $\BB_{T}^\top \WW_{T}\BB_{T} \preceq \LL$} \\
		\leq & \theta^{-0.5} \delta \len{\LL^\dag \zz}_{\LL}
			= \theta^{-0.5} \delta \sqrt{\zz^\top \LL^\dag \zz} \\
		\leq & \theta^{-0.5} \delta n^{0.5}
			\sqrt{\frac{\zz^\top \LL^\dag \zz}{\zz^\top \zz}}
			\qquad \text{since $\len{\zz}^2 \leq n$} \\
		\leq & \sqrt{2} \theta^{-0.5} \delta n^{2.5} U
			\qquad \text{since $\LL^\dag \leq 2n^4 U^2 \II$.}
	\end{align*}
	Then by the inequality
	$\left| x^2 - y^2 \right| \leq (2|y| + |x-y|)|x-y|$,
	the lhs of~(\ref{eq:SolverError}) is at most
	\begin{align*}
		&\kh{2\len{\LL^\dag \zz}_{\BB_{T}^\top \WW_{T}^{1 / 2}
\kh{ \II - (1 - \theta)\WW_{T}^{1 / 2} \BB_{T}\LL^\dag
\BB_{T}^\top\WW_{T}^{1 / 2} }^{-1}
\WW_{T}^{1 / 2} \BB_{T}} + \sqrt{2} \theta^{-0.5} \delta n^{2.5} U} \sqrt{2} \theta^{-0.5} \delta n^{2.5} U
			\\
		\leq &
		\kh{2\theta^{-0.5}\len{\LL^\dag \zz}_{\LL} + \sqrt{2} \theta^{-0.5} \delta n^{2.5} U} \sqrt{2} \theta^{-0.5} \delta n^{2.5} U \\
		& \qquad \text{since $\BB_{T}^\top \WW_{T}^{1 / 2}
\kh{ \II - (1 - \theta)\WW_{T}^{1 / 2} \BB_{T}\LL^\dag
\BB_{T}^\top\WW_{T}^{1 / 2} }^{-1}
\WW_{T}^{1 / 2} \BB_{T} \preceq \frac{1}{\theta} \LL$} \\
		\leq &
		\frac{1}{\theta}\kh{2\sqrt{2}n^{2.5}U + \sqrt{2} \delta n^{2.5} U} \sqrt{2} \delta n^{2.5} U \qquad \text{since $\len{\zz}^2 \leq n$ and $\LL^\dag \preceq 2n^4U^2\II$}\\
		\leq & 6\theta^{-1}\delta n^5 U^2 \qquad
		\text{by $\delta < 1$.}
	\end{align*}
\end{proof}

\begin{proof}[Proof of Lemma~\ref{lem:Lower}]
	\begin{align*}
		& \trace{\LL^\dag \BB_{T}^\top \WW_{T}^{1 / 2}
\kh{ \II - (1 - \theta)\WW_{T}^{1 / 2} \BB_{T}
  \LL^\dag
\BB_{T}^\top\WW_{T}^{1 / 2} }^{-1}
\WW_{T}^{1 / 2} \BB_{T} \LL^\dag} \\
		\geq & \trace{\LL^\dag \BB_{T}^\top \WW_{T} \BB_{T} \LL^\dag} \qquad \text{since $\kh{ \II - (1 - \theta)\WW_{T}^{1 / 2} \BB_{T}
  \LL^\dag
\BB_{T}^\top\WW_{T}^{1 / 2} }^{-1} \succeq \II$} \\
		= &\trace{\LL^\dag \kh{\sum\limits_{e\in T} w(e)\bb_e\bb_e^\top} \LL^\dag} \\
		= & \sum\limits_{e\in T} w(e) \trace{\LL^\dag\bb_e\bb_e^\top \LL^\dag} \\
		\geq & \frac{2\sizeof{T}}{n^2U^2} \qquad \text{by Lemma~\ref{lem:edgeLower} and $w(e) \geq 1$.}
	\end{align*}
\end{proof}